 \documentclass[sigconf]{acmart}
 \renewcommand\footnotetextcopyrightpermission[1]{} 
 
 \usepackage{hyperref}
 \usepackage{amsmath,amssymb}
 \usepackage{stmaryrd}
 \usepackage{paralist}
 \usepackage{graphicx}
 \usepackage{subcaption}
 \usepackage{float}
 \usepackage{amsmath}
 \usepackage{xcolor}
 \usepackage{url}
 \usepackage{balance}
 \usepackage[ruled,vlined,linesnumbered,commentsnumbered]{algorithm2e}
 \usepackage{tikz}
 
 \usetikzlibrary{arrows,trees,positioning,calc}
 \usetikzlibrary{automata,shapes.arrows,backgrounds,fit}
 \usetikzlibrary{decorations.pathmorphing,decorations.pathreplacing}
 \usetikzlibrary{shapes}
 \tikzstyle{rstate}=[state,ellipse]
 \tikzset{>={latex}}

 \newcommand{\dom}{\mathsf{dom}}

 \usepackage{graphics}
 \newcommand{\restrict}{ \scalebox{1}[.85]{\raisebox{.9em}
 		{\mbox{\rotatebox{270}{$\leftharpoonup$}}} }}
 
 \newcommand{\emptyword}{\epsilon}
 \newcommand{\ror}{\vee}
 \def\wild{^*}
 \newcommand{\bind}[2]{\mathop{#1\{#2\}}}
 
 \newcommand{\e}[1]{\emph{#1}}
 \newcommand{\eat}[1]{}
 \newcommand{\doc}{\mathbf{d}}

 \newcommand{\alphabet}{\Sigma}
 \newcommand{\alphabetstar}{\alphabet^{*}}
 
 \newcommand{\mspan}[2]{\ensuremath{[#1,#2\rangle}}
 \newcommand{\allspans}{\mathsf{spans}}

 \newcommand{\join}{\bowtie}
 
 \newcommand{\true}{\mathsf{t}}
 \newcommand{\false}{\mathsf{f}}

 \newcommand{\rc}[1]{\mathsf{#1}}
 \newcommand{\rgxsubc}{\rc{R}}
 \newcommand{\rgxc}{\rc{RGX}}
 \newcommand{\funcrgxc}{\rc{funcRGX}}
 \newcommand{\seqrgxc}{\rc{seqRGX}}
 \newcommand{\disfuncrgxc}{\rc{dfuncRGX}}
 
 \newcommand{\VA}{\textrm{VA}}

 \newcommand{\repspnrpar}[1]{\mathord{\llceil{#1}\rrceil}}
 \newcommand{\repspnr}{\repspnrpar} 
 \newcommand{\reprelational}[1]{\mathord{\llbracket{#1}\rrbracket}}
 \newcommand{\prepspnr}[1]{[{#1}]}

 \newcommand{\spn}[1]{[#1\rangle}

 \newcommand{\vars}{\mathsf{Vars}}

 \newcommand{\ops}{\mathcal{O}}

 \newcommand{\diff}{\setminus}
 
 \newcommand{\df}{:=}

 \def\s#1{\texttt{#1}}
 \def\ititle#1{\textbf{#1}\,\,}
 \def\set#1{\mathord{\{#1\}}}

 \newcommand{\vop}[1]{\mathop{#1{\vdash}}}
 \newcommand{\vcl}[1]{\mathbin{{\dashv}#1}}
 \newcommand{\xalphabet}{\boldsymbol{\Gamma}}

 \newtheorem{theorem}{Theorem}[section]
 \newtheorem{corollary}[theorem]{Corollary}
 \newtheorem{proposition}[theorem]{Proposition}
 
 \newtheorem{lemma}[theorem]{Lemma}

 \newenvironment{citedtheorem}[1]
 {\begin{theorem}\hskip-0.2em\e{\cite{#1}}\,\,}
 	{\end{theorem}}

 \newtheorem{examplethm}[theorem]{Example}
 \newenvironment{example}{\begin{examplethm}\em}
 	{\qed\end{examplethm}}
 
 \newenvironment{citemize}
 {\begin{compactitem}}
 	{\end{compactitem}}
 
 \newenvironment{cenumerate}
 {\begin{compactenum}}
 	{\end{compactenum}}

 \def\Wone{\mathrm{W[1]}}
 \def\Ptime{\mathrm{P}}
 \def\NP{\mathrm{NP}}
 \def\FPT{\mathrm{FPT}}

 \newcommand{\varconf}{{c}}
 \newcommand{\varconfseq}{\tilde{c}}
 \newcommand{\vc}[1]{\mathsf{#1}}

 \newcommand{\el}{\hookleftarrow}
 
 \newcommand{\regex}[1]{\alpha_{\mathsf{#1}}}
 \newcommand{\exvar}[1]{x_{\textit{#1}}}
 \newcommand{\spnrex}[1]{P_{\textit{#1}}}
 \newcommand{\visiblespace}{\text{\textvisiblespace}}

 \newenvironment{repeatresult}[2]
 {\vskip0.5em\par\textsc{#1 #2.}\em}
 {\vskip1em}
 
 \newenvironment{repproposition}[1]{\begin{repeatresult}{Proposition}{#1}}{\end{repeatresult}}
 \newenvironment{reptheorem}[1]{\begin{repeatresult}{Theorem}{#1}}{\end{repeatresult}}
 
 \newenvironment{replemma}[1]{\begin{repeatresult}{Lemma}{#1}}{\end{repeatresult}}

\title{Complexity Bounds for Relational Algebra \\
  over Document Spanners}

\author{Liat Peterfreund}
\affiliation{\institution{Technion}
\city{Haifa}
\country{Israel}}

\author{Dominik D. Freydenberger}
\affiliation{\institution{Loughborough University}
\city{Loughborough} 
\country{United Kingdom}}

\author{Benny Kimelfeld}
\affiliation{\institution{Technion}
\city{Haifa}
\country{Israel}}

\author{Markus Kr{\"{o}}ll}
\affiliation{\institution{TU Wien}
\city{Vienna}
\country{Austria}}

\begin{document}

\begin{abstract}
  We investigate the complexity of evaluating queries in Relational
  Algebra (RA) over the relations extracted by regex formulas (i.e.,
  regular expressions with capture variables) over text documents.
  Such queries, also known as the regular document spanners, were
  shown to have an evaluation with polynomial delay for every positive
  RA expression (i.e., consisting of only natural joins, projections
  and unions); here, the RA expression is fixed and the input consists
  of both the regex formulas and the document. In this work, we
  explore the implication of two fundamental generalizations. The
  first is adopting the ``schemaless'' semantics for spanners, as
  proposed and studied by Maturana et al. The second is going beyond
  the positive RA to allowing the difference operator.

  We show that each of the two generalizations introduces
  computational hardness: it is intractable to compute the natural
  join of two regex formulas under the schemaless semantics, and the
  difference between two regex formulas under both the ordinary and
  schemaless semantics. Nevertheless, we propose and analyze syntactic
  constraints, on the RA expression and the regex formulas at hand,
  such that the expressive power is fully preserved and, yet,
  evaluation can be done with polynomial delay. Unlike the previous
  work on RA over regex formulas, our technique is not (and provably
  cannot be) based on the static compilation of regex formulas, but
  rather on an ad-hoc compilation into an automaton that incorporates
  both the query and the document. This approach also allows us to
  include black-box extractors in the RA expression.
\end{abstract}

\maketitle

\section{Introduction}
The abundance and availability of valuable textual resources position
text analytics as a standard component in data-driven workflows. To
facilitate the integration with textual content, a core operation is
Information Extraction (IE)---the extraction of structured data from
text.  IE arises in a large variety of domains, including biology and
biomedical analysis, social media analysis, cyber
security,\footnote{See, e.g., the TA-COS workshop at
	\url{http://www.ta-cos.org/}.} system and network log analysis, and
business intelligence, to name a
few~\cite{DBLP:journals/ftdb/Sarawagi08,chakraborty2014text}.
\e{Rules} for IE are used in commercial systems and academic
prototypes for text analytics, either as a standalone extraction
language or within machine-learning models.  A common paradigm for
rule programming is the one supported by IBM's
SystemT~\cite{DBLP:conf/acl/LiRC11,DBLP:conf/naacl/ChiticariuDLRZ18},
which exposes a collection of \e{atomic} (sometimes called
``primitive'') extractors of relations from text (e.g., tokenizer,
dictionary lookup, part-of-speech tagger and regular-expression
matcher), together with a relational algebra for manipulating these
relations. In Xlog~\cite{DBLP:conf/vldb/ShenDNR07}, user-defined
functions provide the atomic extractors, and Datalog is used for
relational manipulation. In
DeepDive~\cite{DBLP:journals/sigmod/SaRR0WWZ16}, rules are used for
generating features that are translated into the factors of a
statistical model with machine-learned parameters.  Feature
declaration combines atomic extractors alongside relational operators
thereof.

\subsubsection*{Document spanners}
In this work, we explore complexity aspects of IE within the framework
of \e{document spanners} (or just \e{spanners} for
short)~\cite{DBLP:journals/jacm/FaginKRV15}.  In this framework, a
\e{document} is a string over a fixed finite alphabet, and a
\e{spanner} extracts from every input document a relation of intervals
within the document. An interval, called \e{span}, is represented by
its starting and ending indices in the document. An example of a
spanner is a \e{regex formula}, which is a regular expression with
capture variables that correspond to the relational attributes.  The
most studied language for specifying spanners is that of the
\e{regular} spanners: the closure of regex formulas under the classic
relational algebra: projection, natural join, union, and
difference~\cite{DBLP:journals/jacm/FaginKRV15}.  Equally expressive
formalisms include non-recursive Datalog over regex
formulas~\cite{DBLP:journals/tods/FaginKRV16} and the \e{variable-set
	automaton} (\e{vset-automaton} for short), which is a
nondeterministic finite-state automaton (NFA) that can open and close
variables while running.

Since the framing of the spanner framework, there has been a
considerable effort to delineate the computational complexity of
spanner evaluation, with a special focus on the regular
representations (regex formulas and vset-automata) of the atomic
extractors. Florenzano et al.~\cite{DBLP:conf/pods/FlorenzanoRUVV18}
studied the data complexity (where the spanner is fixed and the input
consists of only the document), and so did Fagin et
al.~\cite{DBLP:journals/corr/abs-1712-08198} who showed that the
closure of regex formulas under Datalog characterizes the class of
polynomial-time spanners.  Freydenberger et
al.~\cite{DBLP:journals/mst/FreydenbergerH18,DBLP:conf/pods/FreydenbergerKP18}
studied the \e{combined complexity} (where the input consists of both
the query and the document) for conjunctive queries, and unions of
conjunctive queries, over spanners. More recently, Amarilli et
al.~\cite{DBLP:journals/corr/AmarilliBMN18} presented an evaluation
algorithm with tractability properties under both data and combined
complexity; we further discuss this algorithm later on.

For complexity analysis, there are important advantages to yardsticks
that take the atomic extractors (e.g., regex formulas or
vset-automata) as input, rather than regarding them small or
fixed. First, the size of these extractors can be quite large in
practice. Taking examples from
\href{http://regexlib.com/}{RegExLib.com}, each of the regexes for
recognizing the RFC 2822 mailbox format (regexp id 711) and date
format (regexp id 969) uses more than 350 ASCII symbols, and a regex
for identifying US addresses (regexp id 1564) uses more than 2,000
ASCII symbols. Furthermore, automata may be constructed by automatic
(machine-learning) processes that achieve accuracy through the
granularity of the automaton. The paradigm of Artificial Neural
Networks (ANNs) in natural-language processing has motivated the
conversion of ANN models such as \e{recurrent neural networks} and
\e{convolutional neural networks} into
automata~\cite{DBLP:journals/corr/abs-1808-09357,DBLP:journals/nn/OmlinG96,DBLP:conf/icml/WeissGY18},
where the number of states may reach tens of thousands to match the
expressiveness of the numeric
parameters~\cite{DBLP:conf/icml/WeissGY18}.  Another advantage of
regarding the atomic extractors as input is more technical:
polynomial-time combined complexity allows to incorporate
vset-automata of which size may depend on the input document. This
approach allows to establish tractability even if we join with
schemaless spanners that cannot be represented as RA expressions over
regular spanners, such as string
equality~\cite{DBLP:conf/pods/FreydenbergerKP18}.

\subsubsection*
{Schema-based functionality vs schemaless sequentiality}
As defined by Fagin et al.~\cite{DBLP:journals/jacm/FaginKRV15}, the
spanners are \e{schema-based} in the sense that every spanner is
associated with a fixed and finite set $X$ of variables, playing the
roles of \e{attributes} in relational databases, so that every tuple
they extract from a document assigns a value to each variable of
$X$. The regex formulas conform to this property in the sense that
every parse tree contains exactly one occurrence of each variable;
such regex formulas are said to be \e{functional}.  Freydenberger~\cite{DBLP:conf/icdt/Freydenberger17} applied the property
of functionality to vset-automata: a vset-automaton is functional if
every accepting path properly opens and closes every variable exactly
once.

The functionality property can be tested in polynomial time for both
regex formulas~\cite{DBLP:journals/jacm/FaginKRV15} and
vset-automata~\cite{DBLP:journals/mst/FreydenbergerH18}. Moreover,
functional vset-automata generalize functional regex formulas in the
sense that every instance of the former can be transformed in
linear time into an instance of the latter (but not necessarily
the other way around).  Beyond that, functional vset-automata (and
regex formulas) possess various desired tractability
features~\cite{DBLP:conf/pods/FreydenbergerKP18}. First, they can be
evaluated with polynomial delay under combined complexity.  Second,
the \e{natural join} of two functional vset-automata can be compiled
in polynomial time into one functional vset-automaton, and so can the
\e{union} of two vset-automata and the \e{projection} of a
vset-automaton to a subset of its variables. Consequently, every
combination of functional vset-automata can be evaluated with
polynomial delay, as long as this combination is via the \e{positive}
operators of the relational algebra.

More recently, Maturana et al.~\cite{DBLP:conf/pods/MaturanaRV18}
introduced a \e{schemaless} version of spanners that allows for
incomplete extraction from documents, in the spirit of the SPARQL
model~\cite{DBLP:journals/tods/PerezAG09}. There, two extracted tuples
may assign spans to different sets of variables. The analog of
functionality is \e{sequentiality}: a regex formula is sequential is
every parse tree includes \e{at most} one occurrence of every
variable, and a vset-automaton is sequential if every accepting path
properly opens and closes every variable \e{at most} once. Again, in
polynomial time we can test for sequentiality and transform a
sequential regex formula into a sequential vset-automaton; moreover,
sequential vset-automata can be evaluated with polynomial delay under
combined complexity~\cite{DBLP:conf/pods/MaturanaRV18}. In fact, the
aforementioned algorithm of Amarilli et
al.~\cite{DBLP:journals/corr/AmarilliBMN18} enumerates with polynomial
delay under combined complexity, and, under data complexity, with
constant delay following a linear pre-processing of the
document.\footnote{This is the spanner analog of a recent line of work
	on the enumeration complexity of database and string
	queries~\cite{DBLP:conf/csl/BaganDG07,DBLP:conf/icdt/CarmeliK18,DBLP:conf/pods/NiewerthS18,DBLP:conf/icdt/Segoufin13}.}
Since functional vset-automata are also sequential, this algorithm
also applies to the schema-based spanners, and improves upon (and, in
fact, generalizes the applicability of) the constant-delay algorithm
of Florenzano et al.~\cite{DBLP:conf/pods/FlorenzanoRUVV18}.

\subsubsection*
{Contribution}
The state of affairs leaves open two fundamental questions regarding
the combined complexity of query evaluation.
\begin{itemize}
	\item Does the tractability for the positive relational algebra
	generalize from the schema-based case to the schemaless case?
	\item Does the tractability extend beyond the positive operators (in
	either the schema-based or schemaless case)? In particular, can we
	enumerate with polynomial delay the \e{difference} between two
	functional vset-automata?
\end{itemize}
We prove that the answers to both questions are negative. More
specifically, it is NP-complete to determine whether the natural join
of two \e{sequential} regex formulas is nonempty
(Theorem~\ref{thm:joinseqhard}), and it is NP-complete to determine
whether the difference between two given \e{functional} regex formulas
is nonempty (Theorem~\ref{thm:diffrgxNPcomp}).

We formulate various syntactic restrictions that allow to avoid
hardness. In particular, we show that polynomial delay is retained if
we bound the number of common variables between the two operands of
the natural join and difference. For the natural join, we also present
a normal form for schemaless regex formulas and vset-automata, namely
\e{disjunctive functional}, that are more restricted than, yet as
expressive as, their sequential counterparts; yet, the natural join of
two disjunctive-functional vset-automata can be compiled into a
disjunctive-functional vset-automaton in polynomial time (hence,
evaluated with polynomial delay).  

In contrast to the natural join, the tractability of the difference
between vset-automata with a bounded number of common variables
\e{cannot} be established via compilation into a single
vset-automaton. This is due to the simple reason that, in the case of
Boolean spanners, the problem is the same as the difference between
two NFAs, where the compilation necessitates an exponential
blowup~\cite{DBLP:journals/tcs/Jiraskova05}. Nevertheless, we establish the
tractability by transforming the difference into a natural join with a
special vset-automaton that is built ad-hoc for the input document.

In summary, our complexity upper bounds are established in two main
approaches. The first is based on a document-independent compilation
of the input vset-automata (or regex formulas) into a new
vset-automaton. The second is based on a compilation of both the input
vset-automata \e{and} the input document into a new, ad-hoc
vset-automaton. We refer to first approach as \e{static
	compilation} and to the second as \e{ad-hoc compilation}.

We compose our tractability results into more general queries by
proposing a new complexity measure that is specialized to
spanners. Recall that the evaluation problem has three components: the
document, the atomic spanners (e.g., regex formulas), and the
relational algebra that combines the atomic spanners, which we refer
to as the \e{RA tree}. Under \e{combined complexity}, all three are
given as input; under \e{data complexity}, the document is given as
input and the rest are fixed; there is also the \e{expression
	complexity}~\cite{DBLP:conf/stoc/Vardi82} where the document is
fixed and the rest are given as input. We propose the \e{extraction
	complexity}, where the RA tree is fixed, and the input consists of
the document and the atomic spanners (mapped to their corresponding
positions in the RA tree). We present and discuss conditions that cast
the extraction complexity tractable (polynomial-delay evaluation) and
intractable (NP-hard nonemptiness). Interestingly, since the
tractability of an RA tree is based on ad-hoc compilation, we can
incorporate there \e{any} polynomial-time spanner, as long as its
dimension is bounded by a constant.

\vspace*{1.6mm}
\subsubsection*{Organization}
The rest of the paper is organized as follows.  In
Section~\ref{sec:prelims}, we present the basic terminology and
concepts.  We investigate the complexity of the natural-join operator
in Section~\ref{sec:join} and the difference operator in
Section~\ref{sec:diff}. We extend our development to the extraction
complexity  in Section~\ref{sec:exp}, and conclude in
Section~\ref{sec:conclusions}.
To meet space constraints, some of the proof are given in the Appendix. 

\section{Preliminaries}\label{sec:prelims}

We first introduce the main definitions and terminology, mainly from
the literature on document
spanners~\cite{DBLP:journals/jacm/FaginKRV15,DBLP:conf/pods/MaturanaRV18}.
\subsection{Document Spanners}

\paragraph*{Documents and spans}
We fix a finite alphabet $\alphabet$ of symbols.  By a \e{document} or
\e{string} we refer to a finite sequence $\doc=\sigma_1 \cdots
\sigma_n$ over $\Sigma$ (that is, each $\sigma_i$ is in $\Sigma$),
that is, a member of $\alphabetstar$.
The length $n$ of the document $\doc=\sigma_1 \cdots
\sigma_n$ is denoted by $|\doc|$. 
A \e{span} is a pair $\mspan i
j$ of indices $ 1 \le i \le j \le n+1$ that marks a substring of
$\doc$. The term $\doc_{ \mspan i j}$ denotes the substring $\sigma_i
\cdots \sigma_{j-1}$.  Note that $\doc _{\mspan i i}$ is the empty
string, and that $\doc_{ \mspan 1 {n+1}}$ is $\doc$.  Note also that
the spans $\mspan i i $ and $\mspan j j$, where $i\ne j$, are
different objects, even though the substrings $\doc_{ \mspan i i}$ and
$\doc_{\mspan j j}$ are equal.  We denote by $\allspans$ the set of
all spans of all strings, that is, all expressions $\mspan i j$ where
$ 1 \le i \le j$.  By $\allspans(\doc)$ we denote the set of all spans
of $\doc$.

\paragraph*{Schemaless spanners}
We assume a countably infinite set $\vars$ of \e{variables}, and
assume that $\vars$ is disjoint from $\alphabet$ and $\alphabetstar$.
A \e{schemaless (document) spanner} is a function that maps each
document into a finite collection of tuples (referred to as
\e{mappings}) that assign spans to variables. More formally, a
\e{mapping} to $\doc$ is a function $\mu$ from a finite set of
variables, called the \e{domain} of $\mu$ and denoted $\dom(\mu)$,
into $\allspans(\doc)$. 
A schemaless spanner is a function $P$ that
maps every document $\doc$ into a finite set $P(\doc)$ of 
mappings.

For a schemaless spanner $P$ and a document $\doc$, different mappings
in $P(\doc)$ may have different domains. This stands in contrast to
the (schema based) spanners of Fagin et
al.~\cite{DBLP:journals/jacm/FaginKRV15}, where $P$ is such that there
exists a set $V_P$ of variables where every document $\doc$ and
mapping $\mu\in P(\doc)$ satisfy $\dom(\mu)=V_P$; in this case, we may
refer to $P$ as a \e{schema-based spanner}.
\begin{figure*}[t]
	\centering
	{
		\begin{tabular}{c}
			$ \underset{\scriptscriptstyle{1}}{\s{R}} \s{odion}\visiblespace
			\underset{\scriptscriptstyle{8}}{\s{R}}\s{askolnikov} \visiblespace 
			\underset{\scriptscriptstyle{20}}{\s{r}} \s{r@edu.ru} \el 
			\underset{\scriptscriptstyle{30}}{\s{Z}} \s{osimov}\visiblespace
			\underset{\scriptscriptstyle{38}}{\s{6}} \s{222345} \visiblespace
			\underset{\scriptscriptstyle{46}}{\s{m}}\s{ov@edu.ru} \el 
			\underset{\scriptscriptstyle{57}}{\s{P}} \s{yotr}\visiblespace
			\underset{\scriptscriptstyle{63}}{\s{L}}\s{uzhin}\visiblespace
			\underset{\scriptscriptstyle{70}}{\s{6}} \s{225545} \visiblespace 	\underset{\scriptscriptstyle{78}}{\s{l}}\s{uzi@edu.uk} \el \cdots$
		\end{tabular}}
		\vspace{-2.9mm}
		\caption{The input document $\doc_{\mathsf{Students}}$}
		\label{fig:indoc}
	\end{figure*}

	\begin{example}\label{ex:spnrstudinfo} 
		Let $\Gamma$ be the alphabet consists of lowercase and uppercase
		English letters: $a,\cdots,z,A\cdots,Z$; digits: $0,\cdots,9$; and
		symbols: $\visiblespace$ that stands for whitespace, `$.$' and
		`$\mathsf{@}$'.  Let $\Delta =\{\el\}$ where $\el$ stands for end of
		line.  The input document $\doc_{\mathsf{Students}}$ over $\Gamma
		\cup \Delta$ given in Figure~\ref{fig:indoc} holds personal
		information on students. (Some of the positions are marked
		underneath for convenience.)  Each line in the document describes
		information on a student in the following format: first name (if
		applicable), last name, phone number (if applicable) and email
		address.  There are spaces in between these elements.  The
		schemaless document spanner $\spnrex{StudInfo}$ extracts from the
		input document $\doc_{\mathsf{Students}}$ the following set of
		mappings, given in a table for convenience.  \vspace{-3mm}
		\begin{figure}[H]
			\begin{tabular}{rllll}
				&$ \exvar{first}$ & $\exvar{last}$ & $ \exvar{mail}$ & $\exvar{phone}$ \\
				\hline
				$\mu_1:$&  $\mspan{1}{7}$ & $\mspan{8}{19}$ & $\mspan{20}{22}$ & \\
				$\mu_2:$&  $\mspan{30}{37}$ &  & $\mspan{46}{56}$ & $\mspan{38}{45}$ \\
				$\mu_3:$&  $\mspan{57}{62}$ & $\mspan{63}{69}$ & $\mspan{79}{89}$ & $\mspan{70}{78}$ \\
			\end{tabular}
		\end{figure}
		\vspace{-2mm}
		Note that the empty cells in the table stand for undefined. 
		That is, we can conclude, for example, that  $\exvar{last} \notin \dom (\mu_2) $.
	\end{example}

	In the next sections, we discuss different representation languages for schemaless
	spanners. Whenever a schemaless spanner is represented by a
	description $q$, we denote by $\repspnrpar{q}$ the actual
	schemaless spanner that $q$ represents.
	{We are using the
		notation $\repspnrpar{\cdot}$ in order to clearly distinguish the
		schemaless semantics from the schema based semantics of Fagin et
		al.~\cite{DBLP:journals/jacm/FaginKRV15} who use
		$\reprelational{\cdot}$. This distinction is critical in the case of
		the \e{vset-automata} that we define later on.}

	\subsection{Regex Formulas}
	One way of representing a schemaless spanner is by means of a \e{regex
		formula}, which is a regular expression with capture variables, as
	allowed by the grammar
	$$\alpha \df
	\emptyset \mid
	\emptyword \mid \sigma \mid
	(\alpha\ror\alpha) \mid
	(\alpha\cdot\alpha) \mid \alpha\wild \mid \bind{x}{\alpha}
	$$
	where $\sigma\in\alphabet$ and $x\in \vars$. 
	For convenience, we sometimes put regex formulas in parentheses
	and also omit parentheses, as long as the meaning remains clear.
	We denote by
	$\vars(\alpha)$ the set of variables that appear in $\alpha$.  By
	$\rgxc$ we denote the class of regex formulas.
	
	Following Maturana et al.~\cite{DBLP:conf/pods/MaturanaRV18}, we
	interpret regex formulas as schemaless spanners in the following
	manner. The following grammar defines the application of a regex
	formula $\alpha$ on a document $\doc=\sigma_1 \cdots \sigma_n$, where
	the result is a pair $(s,\mu)$ where $s$ is a span of $\doc$ and $\mu$
	is a mapping to $\doc$.
	\begin{itemize}
		\item $\prepspnr{\emptyset} (\doc) \df \emptyset$;
		\item $\prepspnr {\epsilon} (\doc) \df \{ (\mspan{i}{i},\emptyset )
		\mid i=1,\dots,n\}$;
		\item $\prepspnr {\sigma} (\doc) \df \{ (\mspan{i}{i+1},\emptyset ) \mid
		\sigma_i=\sigma\}$;
		\item
		$\prepspnr{x\{ \alpha \}} (\doc)\df \{
		(\mspan{i}{j}, \mu \cup \{x\mapsto\mspan{i}{j}\})  \mid 
		(\mspan{i}{j}, \mu) \in \prepspnr{\alpha}(\doc) $ and 
		$	x\not\in\dom(\mu)
		\}$;
		\item
		$\prepspnr {\alpha_1 \vee \alpha_2} (\doc) \df \prepspnr {\alpha_1 } (\doc) \cup \prepspnr {\alpha_2} (\doc) $;
		\item $\prepspnr {\alpha_1 \cdot \alpha_2} (\doc) \df \{
		(\mspan{i}{j},\mu_1 \cup \mu_2) \mid \exists i' \mbox{ s.t. }
		(\mspan{i}{i^\prime},\mu_1) \in \prepspnr{\alpha_1}(\doc)$,
		$(\mspan{i^\prime}{j},\mu_2) \in \prepspnr{\alpha_2}(\doc)$,
		and $\dom(\mu_1)\cap \dom(\mu_2) = \emptyset \}$;
		\item	
		$\prepspnr {\alpha^*} (\doc) \df
		\bigcup_{i=0}^{\infty}  \prepspnr {\alpha^i} (\doc)$
		where $\alpha^i$ stands for the concatenation of $i$ copies of $\alpha$.
	\end{itemize} 
	
	The result of applying $\alpha$ to $\doc$ is then defined as follows.
	$$
	\repspnrpar{\alpha}(\doc) = 
	\{ \mu \mid  ( \mspan{1}{|\doc|+1},\mu) \in \prepspnr {\alpha} (\doc) \}
	$$
	
	We denote by $\repspnrpar{\rgxc}$ the class of schemaless spanners
	that can be expressed using the regex formulas. Similarly, for every
	subclass $\rgxsubc\subseteq\rgxc$, we denote by
	$\repspnrpar{\rgxsubc}$ the class of spanners expressible by an
	expression in $\rgxsubc$.  
	
	\paragraph*{Syntactic restrictions}
	Fagin et al.~\cite{DBLP:journals/jacm/FaginKRV15} introduced the class
	of regex formulas that are interpreted as schema-based spanners, namely
	the \e{functional} regex formulas. To define functional regex
	formulas, we first use the following inductive definition.  A regex
	formula $\alpha$ is functional \e{for} a set $V\subseteq \vars$ of
	variables if:
	\begin{citemize}
		\item $\alpha \in \Sigma^*$ and $V= \emptyset$;
		\item $\alpha = \alpha_1 \vee \alpha_2$ and each $\alpha_i$ is
		functional for $V$;
		\item $\alpha = \alpha_1 \cdot \alpha_2$ and there exists
		$V_1\subseteq V$ such that $\alpha_1$ is functional for $V_1$ and
		$\alpha_2$ is functional for $V\setminus V_1$;
		\item $\alpha = \alpha_0^{*}$ and $\alpha_0$ is functional for
		$\emptyset$;
		\item $\alpha = x\{ \alpha_0\}$ and $\alpha_0$ is functional for
		$V\setminus \{x\}$.
	\end{citemize}
	Finally, a regex formula $\alpha$ is \e{functional} if it is
	functional for the set $\vars(\alpha)$ of its variables. 
	
	Maturana et al.~\cite{DBLP:conf/pods/MaturanaRV18} pointed at a wider
	fragment of regex formulas, namely the \e{sequential} regex formula,
	that has some desirable properties, as will be discussed later.  A
	regex formula $\alpha$ is sequential if the following conditions hold:
	\begin{citemize}
		\item Every sub-formula of the form $\alpha_1 \cdot \alpha_2$
		satisfies $\vars(\alpha_1) \cap \vars(\alpha_2) = \emptyset$;
		\item Every sub-formula of the form $\alpha^{*}$ satisfies
		$\vars(\alpha) = \emptyset$;
		\item Every sub-formula of the form $x\{\alpha \}$ satisfies
		$x \not \in \vars(\alpha)$.\footnote{We added this restriction to
			the original definition~\cite{DBLP:conf/pods/MaturanaRV18} since
			it was mistakenly omitted, as the authors confirmed.}
	\end{citemize}
	We denote by $\funcrgxc$ and $\seqrgxc$ the classes of functional and
	sequential regex formulas, respectively. Maturana et
	al.~\cite{DBLP:conf/pods/MaturanaRV18} showed that
	$\funcrgxc \subsetneq \seqrgxc$, that is, every functional regex
	formula is sequential, but some sequential regex formulas are \e{not}
	functional, as the next example illustrates.
	
	\begin{example}\label{ex:slspnr}
		Let us define the following regex formulas over the alphabet $\Gamma \cup \Delta$ from Example~\ref{ex:spnrstudinfo}:
		$$\regex{mail} \df  
		\exvar{mail}\{   \gamma \mathtt{@} \gamma \mathtt{.} \gamma   \}
		$$ 
		$$
		\regex{name} \df 
		(\exvar{first}\{ \delta \} \visiblespace  \exvar{last}\{ \delta \} ) \vee (\exvar{last}\{ \delta \})
		$$ 
		$$
		\regex{phone} \df 
		\exvar{phone} \{  \beta^* \} 
		$$		where $\gamma \df (a\vee \cdots \vee z)^*$, 	$\delta \df (A\ror \cdots \ror Z)(a \ror \cdots \ror z)^*$,and $\beta  \df (0 \vee \cdots \vee 9 )^*$.
		Based on the previous regex formulas, we define the regex formula that represents the schemaless spanner $\spnrex{StudInfo}$ from Example~\ref{ex:spnrstudinfo}:
		$$\regex{info} \df
		\Gamma^* \cdot (\epsilon \vee \el) \cdot \regex{name} \cdot \visiblespace \cdot 
		\Big(   ( \regex{phone} \cdot \visiblespace \vee \epsilon ) \cdot  \regex{mail}  \Big) \cdot  \el \cdot \Gamma^*
		$$
		Note that this is regex formula is sequential but not functional  since the variables $\exvar{first}$ and $\exvar{phone}$ are optional.
	\end{example}

	\subsection{Vset-Automata}
	In addition to regex formulas, we use the \e{variable-set automata}
	(abbreviated \e{vset-automata}) for representing schemaless spanners,
	as defined by Maturana et al.~\cite{DBLP:conf/pods/MaturanaRV18} as a
	schemaless adaptation of the vset-automata of Fagin et
	al.~\cite{DBLP:journals/jacm/FaginKRV15}.
	
	A \e{vset-automaton}, $\VA$ for short, is a tuple $(Q,q_0,F,\delta)$,
	where $Q$ is set of \e{states}, $q_0\in Q $ is the \e{initial state},
	$F\subseteq Q$ is the set of \e{accepting states},
	%liat:add
	\footnote{\label{fn:multipleaccept}The original definition by Fagin et al.~\cite{DBLP:journals/jacm/FaginKRV15} used a single accepting state. We can extend this definition to multiple accepting states without changing the expressive power by simulating a single accepting state with epsilon transitions.}
	and $\delta$ is a 
	transition relation consisting of \emph{epsilon transitions} of the form $(q,\epsilon, p)$, \emph{letter transitions} of the form 
	$(q,\sigma, p)$ and \e{variable transitions} of the form $(q,\vop{v},p)$ or $(q,\vcl{v},p)$  where $q,p\in Q$, $\sigma \in \Sigma$, and $v\in \vars$. The symbols $\vop{v}$ and $\vcl{v}$ are special symbols to denote the opening or closing of a variable $v$. 
	We define the set $\vars(A)$ as the set of all variables $v$ that are mentioned in some transition of $A$.
	For every finite set $V\subseteq \vars$ we define the set $\Gamma_V\df
	\{\vop{v},\vcl{v} : v\in V\}$ of \emph{variable operations}.
	A \e{run} $\rho$ over a document $\doc \df
	\sigma_1 \cdots \sigma_n$ is a sequence of the form
	\[(q_0,i_0) \overset{o_1}{\rightarrow} \cdots
	(q_{m-1},i_{m-1})\overset{o_m}{\rightarrow} (q_m,i_m)\] where:
	\begin{itemize}
		\item the $i_j$ are indexes in $\set{1,\dots,n+1}$ such that
		$i_0=1$ and $i_m = n+1$; 
		\item each $o_j$ is in $\Sigma\cup\{\epsilon \} \cup \Gamma_{\vars(A)} $;
		\item $i_{j+1} = i_j$ whenever $o_j \in \Gamma_{\vars(A)}$, and 
		$i_{j+1}= i_j+1$ otherwise;
		\item for all $j>0$ we have $(q_{j-1},o_{j},q_j)\in \delta$.
	\end{itemize}
	A run $\rho$ is called \emph{valid} if for every variable $v$ the following hold:
	\begin{itemize}
		\item $v$ is opened (or closed) at most once;
		\item if $v$ is opened at some position $i$ then it is closed at some position $j$ with $i \le j$;
		\item
		if $v$ is closed at some position $j$ then it is opened at some position $i$ with $i\le j$.
	\end{itemize}
	
	A run is called \emph{accepting} if its last state is an
	accepting state, i.e., $q_m \in F$.  For an accepting and valid run
	$\rho$, we define $\mu_{\rho}$ to be the mapping that maps the
	variable $v$ to the span $\mspan{i_j}{i_{j^\prime}}$ where $o_{i_j} =
	\vop{v}$ and $o_{i_{j^\prime}} = \vcl{v}$.  Finally, the result
	$\repspnrpar{A}(\doc)$ of applying the schemaless spanner represented
	by $A$ on a document $\doc$ is defined as the set of all assignments
	$\mu_{\rho}$ for all valid and accepting runs $\rho$ of $A$ on $\doc$.
	A $\VA$ is called \emph{sequential} if all of its accepting runs are
	valid, and it is called \emph{functional} if each such run also
	include all of its variables $\vars(A)$. 
	Note that sequential VAs corresponds with schemaless
	spanners, whereas functional with complete.
	\begin{example}\label{ex:seqautomaton}
		Let $A$ be the following sequential VA:
		
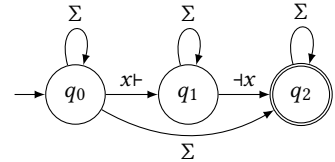
\begin{center}
			\begin{tikzpicture}[on grid, node distance =1.5cm,every loop/.style={shorten >=0pt}]
			\node[state,initial text=,initial by arrow] (q0) {$q_0$};
			\node[state,right= of q0] (q1) {$q_1$};
			\node[state,accepting,right=of q1] (q2) {$q_2$};
			\path[->]
			(q0) edge [loop above] node {$\Sigma$} (q0)
			(q0) edge node[above] {$\vop{x}$} (q1)
			(q1) edge[loop above] node {$\Sigma$} (q1)
			(q1) edge node[above] {$\vcl{x}$} (q2)
			(q2) edge[loop above] node {$\Sigma$} (q2)
			(q0) edge[bend right] node[below] {$\Sigma$} (q2)
			;
			\end{tikzpicture}
		\end{center} 
		Omitting the transition from $q_0$ to $q_2$ results in a functional VA. The same schemaless spanner as that represented by $A$ is given by the sequential regex formula  $\alpha \df (\Sigma^* x\{ \Sigma^* \} \Sigma^*) \vee (\Sigma^+)$ where $\Sigma^+$ stands for $\Sigma \Sigma^*$. 
	\end{example}
	
	\subsection{Algebraic Operators}
	Before we define the algebra over schemaless spanners, we present some
	basic definitions.  Two mappings $\mu_1$ and $\mu_2$ are
	\emph{compatible} if they agree on every common variable, that is,
	$\mu_1(x)=\mu_2(x)$ for all $x\in \dom(\mu_1) \cap \dom(\mu_2)$.  In
	this case, we define $\mu \df \mu_1 \cup \mu_2$ as the mapping with
	$\dom(\mu)=\dom(\mu_1)\cup\dom(\mu_2)$ such that $\mu(x) = \mu_1(x)$
	for all $x\in \dom(\mu_1)$ and $\mu(x) = \mu_2(x)$ for $x\in
	\dom(\mu_2)$.
	
	The correspondents of the relational-algebra operators are defined
	similarly to the SPARQL formalism~\cite{DBLP:journals/tods/PerezAG09}.
	In particular, the operators \e{union}, \e{projection}, \e{natural
		join}, and \e{difference} are defined as follows for all schemaless
	spanners $P_1$ and $P_2$ and documents $\doc$.
	\begin{itemize} 
		\item \ititle{Union:} 
		The union $P \df P_1 \cup P_2$ is defined by
		$P(\doc) \df P_1(\doc) \cup P_2(\doc)$.
		\item \ititle{Projection:} The projection $P \df \pi_Y P_1$ is
		defined by
		$P(\doc)=\set{\mu \restrict Y \mid \mu\in P_1(\doc)}$ where
		$\restrict$ stands for the restriction of $\mu$ to the
		variables in $\dom(\mu)\cap Y$.
		\item \ititle{Natural join:} The \emph{(natural) join}
		$P \df P_1 \join P_2$ is defined to be such that $P(\doc)$
		consists of all mappings $\mu_1 \cup \mu_2$ such that
		$\mu_1\in P_1(\doc)$, $\mu_2\in P_2(\doc)$ and $\mu_1$ and
		$\mu_2$ are compatible.
		\item \ititle{Difference:} The difference
		$P\df P_1 \setminus P_2$ is defined to be such that $P(\doc)$
		consists of all mappings $\mu_1\in P_1(\doc)$ such that no
		$\mu_2\in P_2(\doc)$ is compatible with $\mu_1$.
	\end{itemize}
	We allow the use of these operators for spanners represented by regex formulas or VAs and also for more complex spanner representations, e.g., $\repspnrpar{A_1}\join \repspnrpar{A_2}$.  
	In this case, we use an the abbreviated notation $\repspnrpar{A_1\join A_2}$ instead of $\repspnrpar{A_1}\join \repspnrpar{A_2}$.
	We make the clear note that when the above operators are applied on
	schema-based spanners, they are the same as those of Fagin et
	al.~\cite{DBLP:journals/jacm/FaginKRV15}.

	\begin{example}\label{ex:spnrdiff}
		Let us consider our input document $\doc_{\mathsf{Students}}$ from Figure~\ref{fig:indoc}.
		Assume one wants to filter out from the results obtained by applying the spanner $\spnrex{StudInfo}$ from Example~\ref{ex:slspnr} on $\doc_{\mathsf{Students}}$ the mappings that correspond with students from universities within the UK. 
		It is given that students study in the UK if and only if their email addresses  end with the letters `$\mathsf{uk}$'. 
		We phrase the following regex formula that extracts such email addresses: 
		$$
		\regex{UKm} \df 
		\Big(\epsilon \vee (\Gamma^* \cdot \el) \Big) \cdot
		\Gamma^*  \cdot \visiblespace
		\exvar{mail}\{   \gamma \mathtt{@} \gamma \mathtt{.} \mathsf{uk}  \}
		\cdot  \el \cdot \Gamma^*
		$$ 
		where $\gamma$ is as defined in Example~\ref{ex:slspnr}.
		In this case, the desired output is given by  
		$ \repspnrpar{\regex{info} \setminus \regex{UKm}}(\doc_{\mathsf{Students}})$ who consists of the mappings $\mu_1$ and $\mu_2$ from Example~\ref{ex:spnrstudinfo}.
	\end{example}

	\subsection{Complexity}
	Let $\mathcal{L}$ be a representation language for schemaless spanners
	(e.g., the class of regex formulas or the class of VAs). 
	Given $q \in \mathcal{L}$ and a document $\doc$, we are interested in the decision problem that checks whether $\repspnrpar{q}(\doc)$ is not empty. In that case, we are also interested in evaluating $\repspnrpar{q}(\doc)$.   
	Note that we study the \e{combined complexity} of these
	problems, as both $q$ and $\doc$ are regarded as input.
	
	Under the combined complexity, ``polynomial time'' is not a proper
	yardstick of efficiency for evaluating $\repspnrpar{q}(\doc)$, since this set can contain exponentially many mappings. 
	We thus use
	efficiency yardsticks of
	enumeration~\cite{DBLP:journals/ipl/JohnsonP88}. In particular, our 
	evaluation algorithm takes $q$ and $\doc$ as input, and it outputs all
	the mappings of $\repspnrpar{q}(\doc)$, one by one, without duplicates.
	The algorithm runs in \e{polynomial total time} if its execution time
	is polynomial in the combined size of $q$, $\doc$ and
	$\repspnrpar{q}(\doc)$. 
	The \e{delay} of the evaluation algorithm
	refers to the maximal time that passes between every two consecutive
	mappings. 
	A well-known observation is that polynomial delay implies
	polynomial total time (but not necessarily vice versa), and that
	NP-hardness of the nonemptiness problem implies that no evaluation
	algorithm runs in polynomial total time, or else $\text{P}=\text{NP}$.

	While deciding whether $\repspnrpar{q}(\doc)\neq \emptyset$ is $\NP$-hard whenever $q$ is given as a VA~\cite{DBLP:conf/icdt/Freydenberger17},
	this is not the case
	for sequential (and hence functional) VA:
	\begin{citedtheorem}{DBLP:journals/corr/AmarilliBMN18}\label{thm:enumerationseq}
		Given a sequential VA $A$ and a document $\doc$, 
		one can enumerate
		$\repspnrpar{A}(\doc)$ with 
		polynomial delay.  
	\end{citedtheorem}
	We call two 
	schemaless spanner representations $q_1$ and $q_2$ \e{equivalent} if 
	$\repspnrpar{q_1} \equiv \repspnrpar{q_2}$, that is, $\repspnrpar{q_1}$ and $\repspnrpar{q_2}$ are identical.
	Note that the translation of functional and sequential regex formulas to equivalent functional and sequential VAs, respectively, can be done in linear time~\cite{DBLP:conf/pods/FreydenbergerKP18,DBLP:conf/pods/MaturanaRV18}. Hence, our lower bounds are usually shown for the nonemptiness
	of regex formulas and our upper bounds for 
	the evaluation of
	VAs.

\section{The Natural-Join Operator}\label{sec:join}

To establish complexity upper bounds on the evaluation of schema-based
spanners, Freydenberger et al.~\cite{DBLP:conf/pods/FreydenbergerKP18} used static compilation to
compile the query (where the operands are regex formulas or VAs)
into a single VA.  In particular, they showed that two functional
VAs can be compiled in polynomial time into a single equivalent
VA that is also functional.  Consequently, we can enumerate with
polynomial delay the mappings of $\repspnrpar{A_1 \join A_2}(\doc)$,
given functional VAs $A_1$ and $A_2$.  The question is whether it
generalizes to schemaless spanners: can we efficiently enumerate the
mappings of $\repspnrpar{A_1 \join A_2}(\doc)$, given \e{sequential}
(but not necessarily \e{functional}) $A_1$ and $A_2$? This is no
longer the case, as the next theorem implies, even under the yardstick
of \e{expression complexity}~\cite{DBLP:conf/stoc/Vardi82}
in which the document is regarded as fixed.
(Recall
that a sequential regex formula can be translated in polynomial time
into an equivalent VA~\cite{DBLP:conf/pods/MaturanaRV18}.)
\begin{theorem}\label{thm:joinseqhard}
	The following decision problem is $\NP$-complete. Given two
	sequential regex formulas $\gamma_1$ and $\gamma_2$ and an input
	document $\doc$, is $ \repspnrpar{\gamma_1 \join \gamma_2}(\doc)$
	nonempty? The problem remains $\NP$-hard even if $\doc$ is assumed
	to be of length one.
\end{theorem} 
\begin{proof}
\def\xv#1#2#3{x_{#1}^{#2,#3}}
\def\vv#1#2#3{{#1}^{#2,#3}}
\def\tv{\true}
\def\fv{\false}

Membership in $\NP$ is straightforward, so we focus on NP-hardness. We
show a reduction from 3-CNF-satisfiability which is also known as 3SAT~\cite{Garey:1990:CIG:574848}.  The input for 3SAT
is a formula $\varphi$ with the free variables $x_1,\ldots, x_n$ such
that $\varphi$ has the form $C_1 \wedge \cdots \wedge C_m$, where each
$C_j$ is a clause. In turn, each clause is a disjunction of three
literals, where a literal has the form $x_i$ or $\neg x_i$ for
$i=1,\dots,n$.  The goal is to determine whether there is an
assignment $\tau: \{x_1,\ldots,x_n\} \rightarrow \{0,1\}$ that
satisfies $\varphi$.  Given a 3CNF formula $\varphi$, we construct two
sequential regex formulas $\gamma_1$ and $\gamma_2$ such that there is
a satisfying assignment for $\varphi$ if and only if
$\repspnrpar{\gamma_1\join\gamma_2}(\doc) \ne \emptyset$, where $\doc$
is the document that consists of a single letter $\mathsf{a}$.

To construct $\gamma_1$ and $\gamma_2$, we associate every variable
$x_i$ with $2m$ corresponding capture variables $\xv{i}{j}{\ell}$ for
$1\le j \le m$ and $\ell \in \set{\tv,\fv}$.  We then define
$$\gamma_1 \df \gamma_{x_1} \cdots \gamma_{x_n} \cdot \mathsf{a}$$ where 
$$\gamma_{x_i} \df
( \xv{i}{1}{\tv} \{ \epsilon \}\cdots \xv{i}{m}{\tv} \{ \epsilon \})
\lor (\xv{i}{1}{\fv} \{ \epsilon \}\cdots \xv{i}{m}{\fv} \{ \epsilon
\}).$$ Intuitively, $\gamma_{x_i}$ verifies that the assignment to
$x_i$ is consistent in all of the clauses.  We then define
$$\gamma_2 \df \mathsf{a} \cdot (\delta_1  \cdots \delta_m)$$
where $\delta_j$ is the disjunction of regex formulas $\beta$ such
that $\beta = \xv{i}{j}{\fv}\{\epsilon\}$ if $\neg {x_i}$ appears in
$C_j$, and $\beta = \xv{i}{j}{\tv}\{\epsilon\}$ if $x_i$ appears in
$C_j$.  Intuitively, $\gamma_2$ verifies that at least one disjunct in
each clause is evaluated true.

Let us consider the following example where
$$\varphi \df
(x \lor y \lor z) \wedge (\neg x \lor y \lor \neg z)\,.
$$
In this case, we have
$$
\delta_1 = \vv{x}{1}{\tv}\{\epsilon\} \lor
\vv{y}{1}{\tv}\{\epsilon\} \lor \vv{z}{1}{\tv}\{\epsilon\}
$$ 
$$
\delta_2 = \vv{x}{2}{\fv}\{\epsilon\} \lor
\vv{y}{2}{\tv}\{\epsilon\}
\lor
\vv{z}{2}{\fv}\{\epsilon\}
$$
and, therefore,
\begin{align*}
\gamma_2 \df \mathsf{a} \cdot
(\vv{x}{1}{\tv}\{\epsilon\} \lor \vv{y}{1}{\tv}\{\epsilon\} \lor \vv{z}{1}{\tv)}\{\epsilon\} ) \cdot \\
(\vv{x}{2}{\fv}\{\epsilon\} \lor \vv{y}{2}{\tv}\{\epsilon\}
\lor \vv{z}{2}{\fv}\{\epsilon\} )\,.
\end{align*}
We also have
\begin{align*} 
\gamma_{1} \df
&\big(	\vv{x}{1}{\tv} \{ \epsilon \} \vv{x}{2}{\tv} \{ \epsilon \} 
\lor 
\vv{x}{1}{\fv} \{ \epsilon \}  \vv{x}{2}{\fv} \{ \epsilon \}
\big) \cdot\\
&\big(	\vv{y}{1}{\tv} \{ \epsilon \} \vv{y}{2}{\tv} \{ \epsilon \}
\lor 
\vv{y}{1}{\fv} \{ \epsilon \}  \vv{y}{2}{\fv} \{ \epsilon \}\big)\cdot\\
&\big(	\vv{z}{1}{\tv} \{ \epsilon \} \vv{z}{2}{\tv} \{ \epsilon \}
\lor \vv{z}{1}{\fv} \{ \epsilon \}  \vv{z}{2}{\fv} \{ \epsilon \}\big) \cdot
\mathsf{a}\,.
\end{align*}

It follows directly from the definition that both $\gamma_1$
and $\gamma_2$ are sequential.  Moreover,
$\repspnrpar{\gamma_1 \join \gamma_2}(\doc)$ is nonempty if
and only if there are compatible mappings $\mu_1 \in
\repspnrpar{\gamma_1}(\doc)$ and $\mu_2 \in
\repspnrpar{\gamma_2}(\doc)$.  Since $\gamma_1$ ends with the
letter $\mathsf{a}$ whereas $\gamma_2$ starts with the letter
$\mathsf{a}$, it holds that $\mu_1 \in
\repspnrpar{\gamma_1}(\doc)$ and $\mu_2 \in
\repspnrpar{\gamma_2}(\doc)$ are compatible if and only if
$\dom(\mu_1)\cap \dom(\mu_2) = \emptyset$. We will show that
$\repspnrpar{\gamma_1 \join \gamma_2}(\doc)$ is nonempty if
and only if there is a satisfying assignment to $\varphi$.

\paragraph{The ``only if'' direction}
Suppose that $\repspnrpar{\gamma_1 \join \gamma_2}(\doc)$ is
nonempty.  In this case, a satisfying assignment $\tau$ to
$\varphi$ is encoded by the domain of $\gamma_2$ in the
following way: if $\xv{i}{j}{\ell} \in \dom(\mu_2)$ then
$\tau(x_i) = \ell$. Observe that $\tau$ is well defined, due
to the definition of $\gamma_2$.  

In our example, the mapping $\mu_1 \in
\repspnrpar{\gamma_1}(\mathsf{a})$ with
$$\dom(\mu_1)=\{  \vv{x}{1}{\tv} ,\vv{x}{2}{\tv} ,\vv{y}{1}{\fv} ,\vv{y}{2}{\fv} ,\vv{z}{1}{\fv} ,\vv{z}{2}{\fv}    \}$$
and the mapping $\mu_2 \in \repspnrpar{\gamma_2}(\mathsf{a})$ with  
$$\dom(\mu_2)=\{\vv{x}{1}{\fv} ,\vv{x}{2}{\fv} ,\vv{y}{1}{\tv} ,\vv{y}{2}{\tv} ,\vv{z}{1}{\tv} ,\vv{z}{2}{\tv}    \}$$
are compatible, and the satisfying assignment $\tau$ is
encoded by $\dom(\mu_2)$ and is given by $\tau(x)= \fv$,
$\tau(y) = \tv$ and $\tau(z) = \tv$.

\paragraph{The ``if'' direction}
If there is a satisfying assignment $\tau$ to $\varphi$, then
define the mappings $\mu_1 \in \repspnrpar{\gamma_1}(\doc)$
and $\mu_2 \in \repspnrpar{\gamma_2}(\doc)$ are defined by
$\xv{i}{j}{\ell}\in \dom(\mu_2)$ whenever $j = \tau(x_i)$ and
$\xv{i}{j}{\ell}\in \dom(\mu_2)$ whenever $j \ne
\tau(x_i)$. These mapping are compatible, since
$\dom(\mu_1)\cap \dom(\mu_2) = \emptyset$.  We conclude that
$\repspnrpar{\gamma_1 \join \gamma_2}(\doc)$ is nonempty.

We conclude the $\NP$-hardness of the problem of determining
whether $\repspnrpar{\gamma_1 \join \gamma_2}(\doc)$ is
nonempty, as claimed.

\end{proof}

In what follows, we suggest two different approaches to deal with this
hardness.

\subsection{Bounded Number of Shared Variables}\label{subsec:joincommonvars}
We now consider the task of computing $\repspnrpar{A_1\join A_2}
(\doc)$, given sequential VAs $A_1$ and $A_2$ and a document $\doc$.
Next, we show that compiling the join into a new sequential VA is
Fixed Parameter Tractable (FPT) when the parameter is the number of
common variables.

\def \thmfptjoin {The following problem is FPT when parametrized by $
	|\vars({A_1})\cap \vars({A_2})|$. Given two sequential VAs $A_1$ and
	$A_2$, construct a sequential VA that is equivalent to ${A_1 \join
		A_2}$.}

\begin{lemma}~\label{lem:fptjoin}
	\thmfptjoin
\end{lemma}
Since we have a polynomial delay algorithm for the evaluation of
sequential VAs 
(Theorem~\ref{thm:enumerationseq})
and the size of the resulting VA
is FPT in $ |\vars({A_1})\cap \vars({A_2})|$, we have the following
immediate conclusion.
\begin{theorem}~\label{thm:joinseqfixcommon}
	Given two sequential VAs $A_1$ and $A_2$ and a document $\doc$, one
	can evaluate $\repspnrpar{A_1\join A_2} (\doc)$ with FPT delay
	parameterized by $ |\vars({A_1})\cap \vars({A_2})|$.
\end{theorem}

In the rest of this section, we discuss the proof of
Lemma~\ref{lem:fptjoin}.  
As was shown by Freydenberger et
al.~\cite{DBLP:conf/pods/FreydenbergerKP18} if $A$ is a
functional VA then for every state $q$ of $A$ and every variable
$v \in \vars(A)$, all of the possible runs from the initial state $q_0$ to $q$ include the same variable operations.
Formally, for every state $q$ there is a function $\varconf_q$, namely \emph{the variable configuration function}, that assigns a label from $\{\vc{o}, \vc{c},
\vc{w} \}$, standing for ``open,'' ``close,'' and ``wait,'' to every
variable in $\vars(A)$, as follows. 
First, $\varconf_q(x) = \vc{o}$ if
every run from $q_0$ to $q$ opens $x$ but does not close it. Second,
$\varconf_q(x) = \vc{c}$ if every run from $q_0$ to $q$ opens and
closes $x$. Third, $\varconf_q(x) = \vc{w}$ if no run from $q_0$ to
$q$ opens or closes  variable $x$.

In sequential VAs, however, not all of the accepting runs open and close all of the variables and therefore it makes more sense to replace the label $\vc{w}$ with the label $\vc{u}$ that stands for ``unseen''. 
In addition, in sequential VAs as opposed to functional, there might be a state $q$ for which there are two (different) runs from $q_0$ to $q$ such that the first opens and closes the variable $x$ whereas the second does not even open $x$. For this case, we add to the set of labels the label $\vc{d}$ that stands for ``done'' meaning that variable $x$ cannot be seen after reaching state $q$. Hence, ``done'' can also be understood as ``unseen or closed, depending on what happened before''. We formalize these notions right after the next example.
\begin{example}
	Let us examine the following two accepting runs of the sequential VA $A$ from Example~\ref{ex:seqautomaton} on the input document $\doc \df \mathtt{a}$:
	$$\rho_1 \df (q_0,1) \overset{\vop{x} }{\rightarrow} (q_1,1)  \overset{\mathtt{a} }{\rightarrow} (q_1,2) \overset{\vcl{x} }{\rightarrow} (q_2,2)$$
	$$\rho_2 \df (q_0,1)  \overset{\mathtt{a} }{\rightarrow} (q_2,2)$$
	The run $\rho_1$ gets to state $q_2$ after opening and closing $x$ while $\rho_2$ gets to $q_2$ without opening $x$. Thus, in state $q_2$ the variable configuration of $x$ is $\vc{d}$.  
\end{example}

This ``nondeterministic'' behavior of sequential VAs is reflected in an \emph{extended variable configuration function} $\varconfseq_q$ for every state $q$ whose co-domain is the set $\{\vc{u}, \vc{o},\vc{c}, \vc{d} \}$.
Since all of the accepting runs of a sequential VA are valid, given a state $q$, exactly one of the following holds:
\begin{itemize}
	\item 
	all runs from $q_0$ to $q$ open $x$; in this case $\varconfseq_q(x) = \vc{o}$;
	\item
	all runs from $q_0$ to $q$ (open and) close $x$; in this case  $\varconfseq_q(x) = \vc{c}$; 
	\item
	all runs from $q_0$ to $q$ do not open $x$; in this case $\varconfseq_q(x) = \vc{u}$;
	\item
	at least one run from $q_0$ to $q$ (opens and) closes $x$ and at least one does not open $x$; in this case $\varconfseq_q(x) = \vc{d} $.
\end{itemize}

A sequential VA $A$ is \e{semi-functional for} $x$, if for every state $q$ it holds that $\varconfseq_q(x) \in \{\vc{o}, \vc{c}, \vc{u} \}$. 
We say that $A$ is \e{semi-functional for} $X$ if it is semi-functional for every $x\in X$.
\begin{example}\label{ex:equivseqsemifunc}
	The sequential VA $A$ from Example~\ref{ex:seqautomaton} is not
	semi-functional for $x$ because $\varconfseq_{q_2}(x) = \vc{d}$, as
	reflected from the runs $\rho_1$ and $\rho_2$ presented in the
	previous example.  However, the following equivalent sequential VA
	$A^\prime$ is semi-functional for $x$:
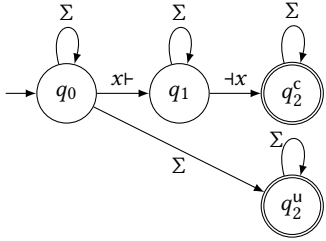
	\begin{center}
		\begin{tikzpicture}[on grid, node distance =1.5cm,every loop/.style={shorten >=0pt}]
		\node[state,initial text=,initial by arrow] (q0) {$q_0$};
		\node[state, right= of q0] (q1) {$q_1$};
		\node[state,accepting,right=of q1] (q2) {$q^{\vc{c}}_2$};
		\node[state,accepting,below=of q2] (q3) {$q^{\vc{u}}_2$};
		\path[->]
		(q0) edge [loop above] node {$\Sigma$} (q0)
		(q0) edge node[above] {$\vop{x}$} (q1)
		(q1) edge[loop above] node {$\Sigma$} (q1)
		(q1) edge node[above] {$\vcl{x}$} (q2)
		(q2) edge[loop above] node {$\Sigma$} (q2)
		(q3) edge[loop above] node [left] {$\Sigma$} (q3)
		(q0) edge[right] node[below] {$\Sigma$} (q3)
		;
		\end{tikzpicture}
	\end{center} 
	Observe that the ambiguity we had in state $q_2$ of $A$ is resolved since it is replaced with two states, each corresponding to a unique configuration.
\end{example}

We show that for every sequential VA $A$, every state $q$ of $A$ and
every variable $v$, we can compute $\varconfseq_q(v)$ efficiently, and
based on that we can translate $A$ into an equivalent sequential VA
that is semi-functional for $X$. We show that the total runtime is FPT
parameterized by $|X|$.

\def\lemsectosemi{Given a sequential VA $A$ and $X \subseteq \vars(A)$, one can construct  in $O(2^{|X|}(n+m))$ time a sequential VA $A^\prime$
	that is equivalent to $A$ and semi-functional for $X$
	where $n$ is the number of states of $A$ and $m$ is the number of its transitions.}
\begin{lemma}\label{lem:seqtosemifunc}
	\lemsectosemi
\end{lemma}
\begin{example}
	The sequential VA $A^\prime$ from Example~\ref{ex:equivseqsemifunc} can be obtained from the 
	automaton $A$ from Example~\ref{ex:seqautomaton} by
	replacing $q_2$ with two states $q^{\vc{u}}_2$ and
	$q^{\vc{c}}_2$ such that $q^{\vc{u}}_2$ corresponds with the
	paths in from $q_0$ to $q_2$ in which variable $x$ was unseen and  $q^{\vc{c}}_2$ corresponds with the
	paths in from $q_0$ to $q_2$ in which variable $x$ was closed, 
	and by changing the
	transitions accordingly.
	The algorithm from the previous Lemma generalizes this idea.
\end{example}
We refer the reader to Footnote~\ref{fn:multipleaccept} in the
definition of a VA and note that, as in the previous example, there are
cases where, to be semi-functional, a VA must have more than a single
accepting state.

If two sequential VAs are semi-functional for their common variables, their  join can be computed efficiently:
\def\lemjoinseqsemi{Given two sequential VAs $A_1$ and $A_2$ that are semi-functional for $\vars(A_1)\cap \vars(A_2)$
	%, each with $O(n)$ states and $O(v)$ variables, one can construct in  $O(vn^4)$ 
	one can construct in polynomial time a sequential VA $A$ that is semi-functional for $\vars(A_1)\cap \vars(A_2)$ and equivalent to ${A_1 \join A_2}$.}
\begin{lemma}\label{lem:joinseqsemi}
	\lemjoinseqsemi
\end{lemma}
The proof of this Lemma uses the same product construction as that for functional VAs presented by Freydenberger et al.~\cite[Lemma 3.10]{DBLP:conf/pods/FreydenbergerKP18}.
What allow us to use the same construction is (a) the fact it ignores the non-common variables and (b) the fact we can treat both $A_1$ and $A_2$ as functional VAs over $\vars(A_1)\cap \vars(A_2)$.

We can now move to compose the proof of Lemma~\ref{lem:fptjoin}:
Given two sequential VAs $A_1$ and $A_2$, we invoke the algorithm from
Lemma~\ref{lem:seqtosemifunc} and obtain two equivalent sequential VAs
$\tilde{A}_1$ and $\tilde{A}_2$, respectively, such that each
$\tilde{A}_i$ is semi-functional for $\vars(A_1)\cap \vars(A_2)$. 
Then, we use Lemma~\ref{lem:joinseqsemi} to join
$\tilde{A}_1 $ and $\tilde{A}_2$.
Note that the runtime is indeed FPT parametrized by $\vars(A_1)\cap \vars(A_2)$.

\subsection{Restricting to Disjunctive Functional} 
Another approach to obtain a tractable evaluation of the join is by restricting the syntax of the regex formulas while preserving expressiveness. 
A regex formula $\gamma$ is said to be \emph{disjunctive functional} if it is a finite disjunction of functional regex formula $\gamma_1,\ldots, \gamma_n$.  
We denote the class of disjunctive functional regex formulas as $\disfuncrgxc$.

Note that every disjunctive functional regex formula is also sequential. However, the regex formula $
z\{\Sigma^* \} \cdot ( x\{\Sigma^* \} \vee   y \{\Sigma^* \} )  
$ is sequential, yet it is not disjunctive functional. 
It also holds that every functional regex formula is disjunctive functional regex formula with a single disjunct. 
We can therefore conclude that we have the following:
$$
\funcrgxc \subsetneq \disfuncrgxc \subsetneq \seqrgxc
$$	
Note that here we treat the regex formulas as syntactic objects.

Equivalently, a \e{disjunctive functional} VA $A$ is the sequential VA whose states are the disjoint union of the states of a finite number $n$ of functional VAs $A_1,\ldots,A_n$ and whose transitions are those of $A_1,\ldots,A_n$, with the addition of a new initial state $q_0$ that is connected with epsilon transitions to each of the initial states of the $A_i$'s. 
Notice that being disjunctive functional is only a syntactic restriction and not semantic as can be concluded from the following proposition. 
\def\propseqtodis{The following hold:
	\begin{enumerate}
		\item 
		For every sequential regex formula there exists an 
		equivalent disjunctive functional regex formula.
		\item 
		For every sequential  VA there exists an 
		equivalent disjunctive functional VA.
	\end{enumerate} }
	\begin{proposition}\label{prop:seqtodis}
		\propseqtodis
	\end{proposition}
	Since $\funcrgxc$ corresponds with schema-based spanners whereas $\seqrgxc$ with schemaless and due to the previous proposition we can conclude the following: 
	$$
	\repspnrpar{\funcrgxc} 
	\subsetneq \repspnrpar{\disfuncrgxc} =
	\repspnrpar{\seqrgxc} 
	$$
	Note that here we refer to the schemaless spanners represented by the regex formulas.
	
	\begin{example}\label{ex:expblowup}
		Consider the following sequential regex formula:
		$$
		(x_1\{\Sigma^*\} \vee y_1 \{\Sigma^* \})\cdots (x_n\{\Sigma^*\} \vee y_n \{\Sigma^* \})
		$$
		Note that if we want to translate it into an equivalent disjunctive functional regex formula then we need at least one disjunct for each possible combination $z_1\{\Sigma^*\} \cdots z_n \{\Sigma^* \}$ where $z_i \in \{x_i,y_i \}$. This implies a lower bound on the length of the shortest equivalent disjunctive functional regex formula.
		Similarly, let us consider the following sequential VA:
		
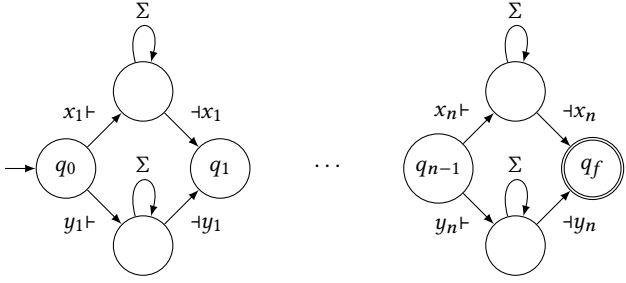
\begin{center}
			\begin{tikzpicture}[on grid, node distance =1.45cm,every loop/.style={shorten >=0pt}]
			\node[state,initial text=,initial by arrow] (q0) {$q_0$};
			\node[state,above right= of q0] (q11) {};
			\node[state,below right=of q11] (q1) {$q_1$};
			\node[state,below right= of q0] (p11) {};
			
			%second block
			\node[above right= of q1] (q22) {};
			\node[below right= of q1] (p22) {};	
			\node[right=of q1] (qdots) {$\cdots$};
			\node[state,right=of qdots] (qnmin1) {$\small{q_{n-1}}$};
			
			\node[state,above right= of qnmin1] (qnn) {};
			\node[state,below right= of qnmin1] (pnn) {};
			\node[state, accepting,below right=of qnn] (qf) {$q_f$};
			
			\path[->]
			(q0) edge node[above left] {$\vop{x_1}$} (q11)
			(q0) edge node[below left] {$\vop{y_1}$} (p11)
			(q11) edge[loop above] node {$\Sigma$} (q11)
			(q11) edge node[above right] {$\vcl{x_1}$} (q1)
			(p11) edge node[below right] {$\vcl{y_1}$} (q1)
			(p11) edge[loop above] node {$\Sigma$} (p11)

			%last
			(qnmin1) edge node[above left] {$\vop{x_n}$} (qnn)
			(qnmin1) edge node[below left] {$\vop{y_n}$} (pnn)
			(qnn) edge[loop above] node {$\Sigma$} (qnn)
			(qnn) edge node[above right] {$\vcl{x_n}$} (qf)
			(pnn) edge node[below right] {$\vcl{y_n}$} (qf)
			(pnn) edge[loop above] node {$\Sigma$} (pnn)
			
			;
			\end{tikzpicture}
		\end{center} 
		An equivalent disjunctive functional VA has at least $2^n$
		accepting states since the states encode the variable
		configurations.
	\end{example}
	
	We record this in the following proposition.
	\def \propexpblowupseqtodisj {For every natural number $n$ the following hold:
		\begin{cenumerate}
			\item 
			There exists a sequential regex formula $\gamma$ of that is the concatenation of $n$ regex formulas of constant length such that each of its equivalent disjunctive functional regex formulas includes at least $2^n$ disjuncts.
			\item
			There exists a sequential VA $A$ with $3n+1$ states such that each of its equivalent disjunctive functional VA has at least $2^n$ states.
		\end{cenumerate}}
		\begin{proposition}~\label{prop:expblowupseqtodisj}
			\propexpblowupseqtodisj
		\end{proposition}
		That is, the translation from sequential to disjunctive functional might necessitate an exponential blow-up.
		Although the translation cannot be done efficiently in the general case, 
		the advantage of using disjunctive functional VAs lies in the fact that we can compile the join of two disjunctive functional VAs efficiently into a disjunctive functional VA.
		\def\joindisj{Given two disjunctive functional VAs $A_1$ and $A_2$, one can construct in polynomial time a disjunctive functional VA $A$ that is equivalent to ${A_1 \join A_2}$.} 
		\begin{proposition}\label{thm:joindisj}
			\joindisj
		\end{proposition} 
		To prove this we can perform a pairwise join between the set of
		functional components of $A_1$ and those of $A_2$ and obtain a set of
		functional VAs for the
		join~\cite[Lemma~3.10]{DBLP:conf/pods/FreydenbergerKP18}.
		
		Since disjunctive functional is a restricted type of sequential VA, we conclude the following.
		\begin{corollary}\label{cor:disjunctive-func-pdelay}
			Given two disjunctive functional VAs $A_1$ and $A_2$ and a input document
			$\doc$, one can enumerate the mappings of
			$\repspnrpar{A_1\join A_2} (\doc)$ in polynomial delay.
		\end{corollary}

\section{The Difference Operator}\label{sec:diff}

When we consider the class of functional VAs,
we know that we can compile all of the positive operators
efficiently (i.e., in polynomial time) into a functional
VA~\cite{DBLP:conf/pods/FreydenbergerKP18}. In the case of
NFAs or regular expressions, compiling the complement into an NFA
necessitates an exponential blowup in
size~\cite{DBLP:journals/jalc/EllulKSW05,DBLP:journals/tcs/Jiraskova05}.
Since NFAs and regular expressions are the Boolean functional
VA and Boolean regex formulas, respectively, we conclude
that constructing a VA that is equivalent to the
difference of two functional VAs, or two functional regex
formulas, entails an exponential blowup.  Therefore, the static compilation
fails to yield tractability results for the difference.  

In the case of NFAs and regular expressions, the membership of a
string in the difference can be tested in polynomial time. In
contrast, the following theorem states that, for functional regex
formulas (and VAs), this is no longer true under the
conventional complexity assumption $\Ptime \ne \NP$.

\def\thmdiffrgxNPcomp{	The following problem is $\NP$-complete. Given two functional regex formulas $\gamma_1$ and $\gamma_2$ with $\vars(\gamma_1) = \vars(\gamma_2)$ and an input document $\doc$, 
	is	$\repspnrpar{\gamma_1 \setminus \gamma_2}(\doc)$ nonempty?  }
\begin{theorem}~\label{thm:diffrgxNPcomp}
	\thmdiffrgxNPcomp
\end{theorem}
\begin{proof}
	Membership in $\NP$ is straightforward: for functional regex formulas, membership can be decided in polynomial time~\cite{DBLP:conf/icdt/Freydenberger17}. Hence, 
	we focus on $\NP$-hardness.  We use a
	reduction from 3SAT as in the proof of
	Theorem~\ref{thm:joinseqhard}. Here, however, we are restricted to
	functional regex formulas and therefore we cannot use the domains of
	the resulting mappings to encode the assignments.  Recall that the
	input is a formula $\varphi$ with the free variables
	$x_1,\ldots,x_n$ such that $\varphi$ has the form $C_1 \wedge \cdots
	\wedge C_m$, where each $C_i$ is a clause.  In turn, each clause is
	a disjunction of three literals, where a literal has the form $x_i$
	or $\neg x_i$.  Given a 3CNF formula, we construct two functional
	regex formulas $\gamma_1$ and $\gamma_2$, and an input document
	$\doc$, such that there is a satisfying assignment for $\varphi$ if
	and only if $\repspnrpar{\gamma_1 \join
		\gamma_2}(\doc)\ne\emptyset$.
	
	We begin with the document $\doc$, which is defined by
	$\doc\df\mathtt{a}^n$. The regex formulas $\gamma_1$ and $\gamma_2$
	are constructed as follows.  We associate every free variable $x_i$
	with a capture variable $x_i$.  We start by defining the auxiliary
	regex formulas $$\beta_i \df ((\bind{x_i}{\emptyword}\cdot
	\mathtt{a})\vee \bind{x_i}{ \mathtt{a}})$$ for $1\le i \le n$ and
	then define
	\begin{equation*}
		\gamma_1\df \beta_1 \cdots \beta_n
	\end{equation*}
	Intuitively, $\gamma_1$ encodes all of the legal assignments for $\varphi$ in such a way that if $x_i$ captures the substring `$\mathtt{a}$' then it corresponds with assigning $\true$ to the free variable $x_i$, and otherwise (in case it captures $\emptyword$), it corresponds with assigning to it $\false$.
	Before defining $\gamma_2$, for each $1 \le i \le m$  we denote the indices of the literals that appear in $C_i$ by $i_1<i_2<i_3$ and define $\gamma_2^i$ as follows:
	\begin{align*}
		\gamma_2^i = 
		&\beta_1\cdots\beta_{i_1-1}\cdot 
		\delta_{i_1}\cdot
		\beta_{i_1+1}\cdots \beta_{i_2-1}\cdot
		\delta_{i_2}\cdot\\
		&\beta_{i_2+1}\cdots\beta_{i_3-1}\cdot
		\delta_{i_3}\cdot
		\beta_{i_3+1}\cdots\beta_{n}
	\end{align*}
	where $\delta_{\ell}$ is defined as $(\bind{x_{\ell}}{\emptyword}\cdot \mathtt{a})$ if $x_{\ell}$ appears as a literal in $C_i$ or as $( \bind{x_{\ell}}{ \mathtt{a}})$ if $\neg x_{\ell}$ appears as a literal in $C_i$
	Intuitively, $\gamma_2^i$ encodes the assignments for which clause $C_i$ is not satisfied.
	We then set $$\gamma_2 \df \bigvee_{1\le i \le m} \gamma_2^i\,.$$

	To emphasize the differences between this reduction and that in the proof of Theorem~\ref{thm:diffrgxNPcomp}, we  consider the same formula:
	$$
	\varphi = (x \vee y \vee z) \wedge ( \neg x \vee y \vee \neg z)
	$$
	We have $\doc \df \mathtt{a}^3$ since we have three variables $\{x,y,z\}$ and 
	$$
	\gamma_1 = \Big((\bind{x}{\emptyword}\cdot \mathtt{a})\vee \bind{x}{ \mathtt{a}}\Big)
	\cdot \Big((\bind{y}{\emptyword}\cdot \mathtt{a})\vee \bind{y}{ \mathtt{a}}\Big) \cdot
	\Big((\bind{z}{\emptyword}\cdot \mathtt{a})\vee \bind{z}{ \mathtt{a}}\Big)
	$$
	For the first clause we have 
	$$
	\gamma_2^1 \df  (\bind{x}{\emptyword}\cdot \mathtt{a})\cdot (\bind{y}{\emptyword}\cdot \mathtt{a}) \cdot (\bind{z}{\emptyword}\cdot \mathtt{a})
	$$
	and for the second 
	$$
	\gamma_2^2 \df  (\bind{x}{ \mathtt{a}})\cdot (\bind{y}{\emptyword}\cdot \mathtt{a}) \cdot (\bind{z}{ \mathtt{a}})
	$$
	It is left to show that $\repspnrpar{\gamma_1 \setminus
		\gamma_2}(\doc) \ne \emptyset$ if and only if $\varphi$ has a satisfying
	assignment.  Note that for every assignment $\mu \in
	\repspnrpar{\gamma_1}(\doc)$ and for every $1 \le j \le n$, it
	holds that $\mu(x_j)$ is either $\mspan{j}{j}$ or
	$\mspan{j}{j+1}$.  Note also that the same is true also for
	$\mu \in \repspnrpar{\gamma_2}(\doc)$.  Let us assume that
	there exists a satisfying assignment $\tau$ for $\varphi$.  We
	define $\mu$ to be the mapping that is defined as
	follows: $\mu(x_i) \df \mspan{i}{i}$ if $\tau(x_i) = \false$
	and $\mu(x_i) \df \mspan{i}{i+1}$, otherwise (if $\tau(x_i) =
	\true$).  It then follows immediately from the definition of
	$\gamma_2$ that $\mu \in \repspnrpar{\gamma_1 \setminus
		\gamma_2}(\doc)$.  On the other hand, assume that $\mu \in
	\repspnrpar{\gamma_1 \setminus \gamma_2}(\doc)$.  We can
	define an assignment $\tau$ is such a way that $\tau(x_i) =
	\true$ if $\mu(x_i) = \mspan{i}{i+1}$ and $\tau(x_i) = \false$
	otherwise (if $\mu(x_i) = \mspan{i}{i}$).  It follows directly
	from the way we defined $\gamma_1$ and $\gamma_2$ that $\tau$
	is a satisfying assignment for $\varphi$.

	In our example, the assignment $\tau$ defined by $\tau(x)=\tau(y) = \true$ and $\tau(z) = \false$ is a satisfying assignment. Indeed, the mapping $\mu$ corresponds to this assignment that is defined by $\mu(x) = \mspan{1}{2}$, $\mu(y) = \mspan{2}{3}$ and $\mu(z) = \mspan{3}{3}$ is in $\repspnrpar{\gamma_1}(\mathtt{a}^n)$ but is not in $\repspnrpar{\gamma_2} (\mathtt{a}^n)$ since either (a) $\mu(x) = \mspan{1}{1}$ and $\mu(y) = \mspan{2}{2}$) or  (b) $\mu(x) = \mspan{1}{2}$ and $\mu(y) = \mspan{2}{2}$.
	Note also that the assignment $\mu$ defined by $\mu(x) = \mspan{1}{2}$, $\mu(y) = \mspan{2}{3}$ and $\mu(z)=\mspan{3}{4}$ is in $\repspnrpar{\gamma_1 \diff \gamma_2}(\mathtt{a}^n)$ since it is in $\repspnrpar{\gamma_1 }(\mathtt{a}^n)$ and not in $\repspnrpar{\gamma_2}(\mathtt{a}^n)$. Indeed, the assignment $\tau$ for which $\tau(x) = \tau(y) = \tau(z)$ is a satisfying assignment for $\varphi$. 
	
	We conclude the $\NP$-hardness of determining the nonemptiness of $\repspnrpar{\gamma_1 \setminus \gamma_2}(\doc)$.
\end{proof}

We can conclude from Theorem~\ref{thm:diffrgxNPcomp} that, in contrast
to the tractability of the natural join of disjunctive functional VAs
(Corollary~\ref{cor:disjunctive-func-pdelay}), here we are facing
$\NP$-hardness already for functional VAs.  In the remainder of this
section, we discuss syntactic conditions that allow to avoid this
hardness.

\subsection{Bounded Number of Common Variables} 
Theorem~\ref{thm:diffrgxNPcomp} implies that no matter what approach
we choose to tackle the evaluation of the difference, without imposing
any restrictions we hit $\NP$-hardness.  In this section, we
investigate the restriction of an upper bound on the number of common
variables shared between the operands. Recall that this restriction
leads to an FPT static compilation for the natural join
(Lemma~\ref{lem:fptjoin}).  Yet, we observed at the beginning of
Section~\ref{sec:diff}, in the case of difference, such static
compilation necessitates an exponential blow-up, even if there are no
variables at all,

Therefore, instead of static compilation that is independent of the
document, we apply an ad-hoc compilation that depends on the specific
document at hand. In this case, we refer to the resulting automaton as
an \emph{ad-hoc VA} since it is valid only for that specific document.
{Ad-hoc VAs} were introduced (without a name) by Freydenberger et
al.~\cite{DBLP:conf/pods/FreydenbergerKP18} as a tool for evaluating
functional VAs with polynomial delay.  The next lemma is based on this
idea.

\def\thmdiffseqkvars{Let $k$ be a fixed natural number.  Given two
	sequential VAs $A_1$ and $A_2$ where $|\vars(A_1) \cap
	\vars(A_2)|\le k$ and a document $\doc$, one can construct in
	polynomial time a sequential VA $A_{\doc }$ with
	$\repspnrpar{A_{\doc}}(\doc) = \repspnrpar{A_1 \setminus A_2}
	(\doc)$.}
\begin{lemma}\label{thm:diffseqkvars}
	\thmdiffseqkvars
\end{lemma}

Since we can enumerate the result of sequential VA with polynomial
delay (Theorem~\ref{thm:enumerationseq}), we can conclude the
following.
\begin{theorem}
	Let $k$ be a fixed natural number. Given two sequential VAs $A_1$
	and $A_2$ where $|\vars(A_1) \cap \vars(A_2)|\le k$ and a document
	$\doc$, one can enumerate $ \repspnrpar{A_1 \setminus A_2}(\doc)$
	with polynomial delay.
\end{theorem}

We now present the proof sketch of Lemma~\ref{thm:diffseqkvars}.  
\begin{proof}[Proof Sketch]
	We construct two sequential VAs $A$ and $B$ (that share a
	bounded number of variables) such that evaluating the difference of
	$A_1$ and $A_2$ on $\doc$ is the same as evaluating the natural join
	of $A$ and $B$ on $\doc$. This natural join can be compiled into a
	sequential VA in polynomial time when the number of common variables
	is bounded by a constant (Theorem~\ref{thm:joinseqfixcommon}), and
	therefore, we establish the desired result.
	
	Yet, unlike the schema-based model, difference in the schemaless
	case cannot be translated straightforwardly into a natural join
	(e.g., via complementation).  For illustration, let us consider the
	case where there are $\mu_1 \in \repspnrpar{A_1} (\doc)$ and $\mu_2
	\in \repspnrpar{A_2} (\doc)$ such that $\dom(\mu_1) \cap \dom(
	\mu_2) = \emptyset$.  In this case, the assignment $\mu_1$ is not in
	$\repspnrpar{A_1 \setminus A_2} (\doc)$ since it is compatible with
	$\mu_2$. Nevertheless, $\mu_1$ will occur in the natural join of
	$A_1$ with every VA $A_2'$, unless $A_1$ and $A_2'$ share one or
	more common variables.
	
	As a solution, we construct a VA that encodes information about the
	domains of the mappings $\mu$, within the variables shared by $A_1$
	and $A_2$, using new shared \e{dummy} variables. Specifically, we
	have a dummy variable $\hat x$ for every shared variable $x$. If
	$x\in\dom(\mu)$, then $\hat x$ is assigned the first empty span
	$\mspan{1}{1}$, and if $x\notin\dom(\mu)$, then $\hat x$ is assigned
	the last empty span $\mspan{|\doc|+1}{|\doc|+1}$.  (Here, we assume
	that $\doc$ is nonempty; we deal separately with the case $\doc =
	\epsilon$.)
	
	We construct a VA $A$ for the above extended mappings of $A_1$.  In
	addition, we construct a VA $B$ by iterating through all possible
	extended mappings over the shared variables, and for each such a
	mapping, if it is incompatible with all of the extended mappings of
	$\repspnr{A_2}(\doc)$, then we include it in $B$. This construction
	can be done in polynomial time, since we assume that the number of
	common variables is bounded by a constant.
	
	We conclude by showing that the extended mappings of
	$\repspnr{A}(\doc)$ that have compatible mappings in
	$\repspnr{B}(\doc)$ correspond to the mappings of
	$\repspnr{A_1}(\doc)$ that have no compatible mappings in
	$\repspnr{A_2}(\doc)$, and also that the extended mappings of
	$\repspnr{A}(\doc)$ that have compatible mappings in
	$\repspnr{B}(\doc)$ correspond to the mappings of
	$\repspnr{A_1}(\doc)$ that do not have compatible mappings in
	$\repspnr{A_2}(\doc)$.  \let\oldqedsymbol\qedsymbol
	\renewcommand{\qedsymbol}{\oldqedsymbol \emph{(Proof sketch)}}{}
\end{proof}

Theorem~\ref{thm:diffseqkvars} shows that we can enumerate the
difference with polynomial delay when we restrict the number of common
variables. A natural question is whether the degree of this polynomial
depends on this number; the next theorem answers this question
negatively, under the conventional assumptions of parameterized
complexity.

\def\thmdiffwonehard{ The following problem is $\Wone$-hard
	parametrized by $|\vars(\gamma_1)\cap \vars(\gamma_2)|$.  Given two
	functional regex formulas $\gamma_1$ and $\gamma_2$ and an input
	document $\doc$, is $\repspnrpar{\gamma_1 \setminus \gamma_2}(\doc)$
	nonempty?}
\begin{theorem}\label{thm:diffwonehard}
	\thmdiffwonehard
\end{theorem}
This result contrasts our $\FPT$ result for the natural join
(Theorem~\ref{thm:joinseqfixcommon}).  The proof uses a reduction
from the problem determining whether a 3-SAT formula has a satisfying
assignment with at most $p$ ones, where $p$ is the
parameter~\cite{downey2012parameterized}.

\subsection{Restricting the Disjunctions}
We now propose another restriction that guarantees a tractable
evaluation, this time allowing the number of common variables to be
unbounded. We begin with some definitions.  

Let $\gamma$ be a sequential regex formula and let $x \in \vars$ be a
variable. Then $\gamma$ is \e{synchronized} for $x$ if, for every
subexpression of $\gamma$ of the form $\gamma_1 \vee \gamma_2$, we
have that $x$ appears neither in $\gamma_1$ nor in $\gamma_2$.  A
regex formula $\gamma$ is called \emph{synchronized} for $X\subseteq
\vars$ if it is synchronized for every $x\in X$.

This notion generalizes to sequential VAs: A state $q$ of a
sequential VA~$A$ is called a \emph{unique target state}
for the variable operation $\omega \in \Gamma_{\vars(A)}$, if for
every state $p$ of $A$ we have that $(p,\omega,q)\in \delta$ implies
$q = q_{\omega}$ where $\delta$ is the transition relation of
$A$. In other words, $ q_{\omega}$ is the only state that can be reached by processing $\omega$.   
We say
that $A$ is \e{synchronized} for a variable $x \in \vars$ if each of
$\vop{x}$ and $\vcl{x}$ has a unique target state and either all
accepting runs of $A$ open and close $x$, or no accepting run of $A$
operates on $x$. 
Finally, $A$
is \e{synchronized} for $X\subseteq \vars$ if it is synchronized for
every $x\in X$.
\begin{example}
	Consider the regex formula $(\bind{x}{\Sigma^*} \vee \epsilon)\cdot \bind{y}{\Sigma^*}$ and this equivalent VA: 
	
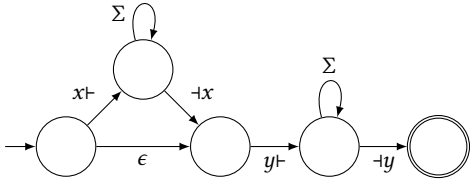
\begin{center}
		\begin{tikzpicture}[on grid, node distance =1.45cm,every loop/.style={shorten >=0pt}]
		\node[state,initial text=,initial by arrow] (q0) {};
		\node[state,above right= of q0] (q11) {};
		\node[state,below right=of q11] (q1) {};
		\node[state,right= of q1] (q2) {};
		\node[state,right= of q2,accepting] (q3) {};
		
		\path[->]
		(q0) edge node[above left] {$\vop{x}$} (q11)
		(q11) edge[loop above] node[left, near start] {$\Sigma$} (q11)
		(q11) edge node[above right] {$\vcl{x}$} (q1)
		(q0) edge node[below] {$\epsilon$} (q1)
		(q1) edge node[below] {$\vop{y}$} (q2)
		(q2) edge[loop above] node {$\Sigma$} (q2)
		(q2) edge node[below] {$\vcl{y}$} (q3)
		;
		\end{tikzpicture}
	\end{center} 
	Both are synchronized for $y$ and not for $x$:
	The regex formula has a subexpression of the form $(x\{ \Sigma^*\} \vee \epsilon) $, whereas the variable $y$ does not appear under any disjunction.
	In the VA, although each variable operation has a unique target state, not all of the accepting runs 
	include the variable operations $\vop{x}$ and $\vcl{x}$ (as opposed to $\vop{y}$ and $\vcl{y}$, which are included in every accepting run).
\end{example}
The following result states that conversions from regex formulas to VAs can preserve the property of being
synchronized for $X$. 
\def\lemsyncregextosyncautomaton{Let
	$\gamma$ be a sequential regex formula that is synchronized for
	$X\subseteq \vars$.  One can convert $\gamma$ in linear time into an
	equivalent sequential VA $A$ that is synchronized for
	$X$.}
\begin{lemma}\label{lem:syncregextosyncautomaton}
	\lemsyncregextosyncautomaton
\end{lemma}
As one might expect, VAs that are synchronized (for some
non\-empty set $X$ of variables) are less expressive than sequential or
semi-functional VAs (that are defined in Section~\ref{subsec:joincommonvars}). In fact, even
functional regex formulas can express spanners that are not
expressible with VAs that are synchronized for all their
variables:

\def\propsyncnotexpressive{Let $\gamma\df (\mathsf{a}\cdot x \{ \epsilon \} \cdot \mathsf{a} ) \vee (\mathsf{b}\cdot x \{ \epsilon \} \cdot \mathsf{b})$. There is no sequential VA $A$ that is synchronized for $x$ and equivalent to $\gamma$.}
\begin{proposition}\label{prop:syncnotexpressive}
	\propsyncnotexpressive
\end{proposition}

Hence, by using synchronized VAs, we sacrifice expressive power. But this restriction also allows us to state the following positive result on the difference of VAs:
\def \thmdiffseqsync {	
	Given an input document $\doc$  and two sequential VAs $A_1$ and $A_2$ such that, for $X\df \vars(A_1) \cap \vars(A_2)$, $A_1$ is semi-functional for  $X$ and
	$A_2$ is synchronized for $X$, one can construct a sequential VA $A_{\doc}$ with $\repspnrpar{A_{\doc}}(\doc) = \repspnrpar{A_1 \setminus A_2} (\doc)$ in polynomial time.
	%in time $O(\ell^2 v^2 n  + \ell v^3 n^3 + n^4)$. Here, $n$ is the number of states of $A_1$ and $A_2$, $\ell=|\doc|$, and $v=|\vars(A_1)|$.  
}
\begin{theorem}\label{thm:thmdiffseqsync}
	\thmdiffseqsync
\end{theorem}
The full proof can be found in the Appendix; we discuss 
some of its key ideas. The first key observation is that $A_2$ can be
treated as a functional VA that uses only the common
variables (similarly to the proof of
Lemma~\ref{lem:joinseqsemi}). This allows us to work with the variable
configurations of $A_2$, and construct the \emph{match structure}
$M(A_2,\doc)$ of $A_2$ on $\doc$. This model was introduced (without a
name) by Freydenberger et al.~\cite{DBLP:conf/pods/FreydenbergerKP18}
to evaluate functional VAs with polynomial delay.  As
explained there, every element of $\repspnrpar{A_2}(\doc)$ can be
uniquely expressed as a sequence of $|\doc|+1$ variable configurations
of $A_2$.

Every accepting run of $A_2$ on $\doc$ can be mapped into such a sequence by taking the variable configurations of the states just before a symbol of $\doc$ is read (and the configuration of the final state). The match structure  $M(A_2,\doc)$ is an NFA that has the set of variable configurations of $A_2$ as its alphabet; and its language is exactly the set of sequences of variables configurations that correspond to elements of $\repspnrpar{A_2}(\doc)$. 

While determinizing match structures is still hard, the fact that $A_2$ is synchronizing on the common variables allows us to construct a deterministic match structure $D_2$ from $M(A,\doc)$. Using a variant of the proof of Lemma~\ref{lem:joinseqsemi}, we can then combine $A_1$ and $A_2$ into an ad-hoc VA $A_\doc$ with $\repspnrpar{A_{\doc}}(\doc) = \repspnrpar{A_1 \setminus A_2} (\doc)$.

After creating $A_\doc$ according to  Theorem~\ref{thm:thmdiffseqsync}, we can use Theorem~\ref{thm:enumerationseq} to obtain the following tractability result:
\begin{corollary}
	Given an input document $\doc$ and two sequential VAs $A_1$ and $A_2$ such that, for $X\df\vars(A_1) \cap \vars(A_2)$, $A_1$ is semi-functional for $X$  and 
	$A_2$ is synchronized for  $X$, one can enumerate the mappings in $\repspnrpar{A_1 \setminus A_2} (\doc)$ in polynomial delay.
\end{corollary}

We saw that disallowing disjunctions over the variables 
leads to tractability.  Can we relax this restriction by allowing a
fixed number of such disjunctions? Our next result is a step towards
answering this question.  A \emph{disjunction-free} regex formula is a
regex formula that does not contain any subexpression of the form
$\gamma_1 \vee \gamma_2$.
\begin{proposition}\label{prop:disjunctionfreeExpensive}
	The following decision problem is $\NP$-complete. 
	Given two sequential regex formulas $\gamma_1$ and $\gamma_2$ with 
	$\vars(\gamma_1) = \vars(\gamma_2)$
	and an input document $\doc$ such that 
	\begin{itemize}
		\item $\gamma_1$ is functional,
		\item $\gamma_2$ is a disjunction of regex formulas $\gamma_2^i$ such that each is disjunction-free,
		\item for every variable $x\in \vars(\gamma_2)$, it holds that $x$ appears in at most 3 disjuncts $\gamma_2^i$ of $\gamma_2$,
	\end{itemize}    
	is $\repspnrpar{\gamma_1 \setminus \gamma_2}(\doc)$ nonempty? 
\end{proposition}
\begin{proof}
	This proof is an adaption of the proof of Theorem~\ref{thm:diffrgxNPcomp},
	using mostly the same notation.
	Instead of a general 3CNF formula, let $\varphi=C_1\wedge\ldots\wedge C_m$
	be a CNF formula, such that every clause
	$C_i$ contains either 2 or 3 literals, and each of the  variables appears in at most
	$3$ clauses. 
	Deciding satisfiability for such a formula is still
	NP-complete~\cite{DBLP:journals/dam/Tovey84}.
	
	For $\gamma_1$ to not have any disjunctions, we first set 
	set $\doc=(\mathtt{bab})^n$ for some $\mathtt{a},\mathtt{b}\in\alphabet$.
	We then define
	\begin{equation*}
		\gamma_1=(\mathtt{b}\bind{x_1}{\mathtt{a}\wild}\cdot\mathtt{a}\wild\mathtt b)
		\cdots(\mathtt{b}\bind{x_n}{\mathtt{a}\wild}\cdot\mathtt{a}\wild\mathtt b).
	\end{equation*}
	Intuitively $\gamma_1$ encodes all of the possible assignments.
	The regex formula $\gamma_2$ is defined analogously to $\gamma_2$ in 
	the proof of Theorem~\ref{thm:diffrgxNPcomp} with an adaptation to the new input document and a slight simplification of the $\gamma_2^i$s (since we do not
	need $\gamma_2$ to be functional any more). 
	Formally, we set $$\gamma_2^i= (\mathtt{bab})^{i_1-1} \delta_{i_1}(\mathtt{bab})^{i_2-i_1-1}\delta_{i_2}(\mathtt{bab})^{n-i_2}$$ if only variables $x_{i_1},x_{i_2}$ with $i_1<i_2$ appear in clause $C_i$, and $$\gamma_2^i= (\mathtt{bab})^{i_1-1} \delta_{i_1}(\mathtt{bab})^{i_2-i_1-1}\delta_{i_2}(\mathtt{bab})^{i_3-i_2-1}\delta_{i_3}(\mathtt{bab})^{n-i_3}$$ if variables $x_{i_1},x_{i_2},x_{i_3}$  with $i_1<i_2<i_3$ appear in clause $C_i$.
	
	By the choice of the 3CNF formula $\varphi$, every variable $x_j$ appears in at most
	three regex formulas of the form $\gamma_2^i$. 
	Correctness of this reduction
	can be shown analogously to that of Theorem~\ref{thm:diffrgxNPcomp}.
\end{proof}

We conclude that evaluating $\gamma_1 \setminus \gamma_2$ remains hard even if $\gamma_1$ is functional (and hence also semi-functional for the common variables)  and $\gamma_2$ is a disjunction of disjunction-free regex formulas, and each of $\gamma_2$'s variables appears in at most three such disjuncts. 

It is open whether the problem becomes tractable if the variables are limited to at most one or two disjuncts.

\section{Extraction Complexity}\label{sec:exp}

\def\L{\mathcal{L}}

In this section, we discuss queries that are defined as RA expressions
over schemaless spanners given in a representation language $\L$
(e.g., regex formulas), which we refer to as the language of the
\e{atomic} spanners.  Formally, an \emph{RA tree} is a directed and
ordered tree whose inner nodes are labeled with RA operators, the
out-degree of every inner node is the arity its RA operator, and each
of the leaves is a placeholder for a schemaless spanner.  For
illustration, Figure~\ref{fig:ratree1} shows an RA tree $\tau$, where
the placeholders are the rectangular boxes with the question marks;
the dashed arrows should be ignored for now. The RA tree corresponds
to the relational concept of a \e{query tree} or a \e{logical query
	plan}~\cite{DBLP:books/ph/Garcia-MolinaUW99,DBLP:journals/cacm/SmithC75}.
As in the rest of the paper, we restrict the discussion to the RA
operators projection, union, natural join, and difference.

Let $\L$ be a representation language for atomic spanners, and let
$\tau$ be an RA tree.  An \e{instantiation} of $\tau$ assigns a
schemaless spanner representation from $\L$ to every placeholder, and
a set of variables to every projection. For example,
Figure~\ref{fig:ratree1} shows an instantiation $I$ for $\tau$ via the
dashed arrows; here, we can think of $\L$ as the class of sequential
regex formulas, and so, each $\alpha$ expression is a sequential regex
formula. 

An instantiation $I$ of $\tau$ transforms $\tau$ into an actual
schemaless spanner representation, where $\tau$ is the parse tree of
its algebraic expression. We denote this representation by
$I[\tau]$. As usual, by $\repspnrpar{I[\tau]}$ we denote the actual
schemaless spanner that $I[\tau]$ represents.

\newcommand{\Ltree}{\mathcal{L}}
\newcommand{\instfun}{\iota}
\newcommand{\treeop}{\xi} 
\newcommand{\extrtree}{\tau}
\newcommand{\insttree}[1]{\extrtree^{#1}}
\newcommand{\opstree}{\ops_{\extrtree}}

\begin{example}\label{ex:sec5studrec}
	Assume that the input document $\doc_{\mathsf{Students}}$ from the
	earlier examples is now extended and contains additional information
	about the students, including recommendations they got from their
	professors and previous hires.  Let us assume that every line begins
	with a student's name and contains information about that student.
	Let us also assume that we have the following functional regex
	formulas:
	\begin{itemize}
		\item regex formula $\regex{sm}$  with capture variables $\exvar{stdnt}, \exvar{ml}$ that extracts  names with their corresponding email addresses;
		\item
		regex formula $\regex{sp}$  with  variables $\exvar{stdnt}, \exvar{phn}$ that extracts names with their corresponding phone numbers;
		\item
		regex formula $\regex{nr}$  with  variables $\exvar{stdnt},
		\exvar{rcmnd}$ that extracts names with their corresponding recommendations.
	\end{itemize} 
	Note that all of the regex formulas are functional, that is, they do
	not output partial mappings. The following query extracts the
	students that not have recommendations.
	$$
	\pi_{\{\exvar{stdnt}\}}  \Big( (\regex{sm} \join \regex{sp}) \setminus (\regex{nr}) \Big)
	$$
	This query is $I[\tau]$ for the RA tree $\tau$ and the instantiation
	$I$
	of Figure~\ref{fig:ratree1}.  This query
	defines the spanner $\repspnrpar{I[\tau]}$, and the set of extracted
	spans is $\repspnrpar{I[\tau]}(\doc_{\mathsf{Students}})$.
\end{example}

\begin{figure}[t]
	\begin{center}
		\input{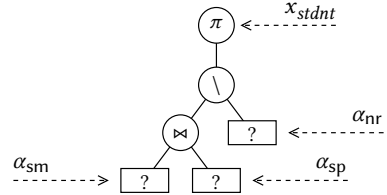}
	\end{center}
	\vskip-1em
	\caption{\label{fig:ratree1}An RA tree $\tau$ with an instantiation
		$I$}
	\vskip-1em
\end{figure}

We present a complexity measure that is unique to spanners, namely the
\e{extraction complexity}, where the RA tree $\tau$ is regarded fixed
and the input consists of both the instantiation $I$ and the input
document $\doc$. Specifically, the \e{evaluation problem} for an RA
tree $\tau$ is that of evaluating $\repspnrpar{I[\tau]}(\doc)$, given
$I$ and $\doc$. Similarly, the \e{nonemptiness problem} for an RA tree
$\tau$ is that of deciding whether $\repspnrpar{I[\tau]}(\doc)$ is
nonempty, given $I$ and $\doc$.

Clearly, some RA trees have an intractable nonemptiness and,
consequently, an intractable evaluation. For example, if $\L$ is the
class of sequential regex formulas and $\tau$ is the RA tree that
consists of a single natural-join node, then the nonemptiness problem
for $\tau$ is NP-complete (Theorem~\ref{thm:joinseqhard}). Also, if
$\L$ is the class of functional regex formulas and $\tau$ is the RA
tree that consists of a single difference node, then the nonemptiness
problem for $\tau$ is NP-complete (Theorem~\ref{thm:diffrgxNPcomp}).
In contrast, by composing the positive results established in
Sections~\ref{sec:join} and~\ref{sec:diff}, we obtain the following
theorem, which is a consequence of Lemma~\ref{lem:fptjoin} and
Lemma~\ref{thm:diffseqkvars}.
\begin{theorem}\label{thm:extractioncomp}
	Let $\Ltree$ be the class of sequential VAs. Let $k$ be a fixed
	natural number and $\extrtree$ an RA tree.  The evaluation problem
	for $\extrtree$ is solvable with polynomial delay, assuming that for
	all join and difference nodes $v$ of $I[\tau]$, the left and right
	subtrees under $v$ share at most $k$ variables.
\end{theorem}
We restate that, while static compilation suffices for the positive
operators, we need ad-hoc compilation to support the difference.
Interestingly, the ad-hoc approach allows us to incorporate into the
RA tree other representations of schemaless spanners, which can be
treated as black-box schemaless spanners, as long as these spanners
can be evaluated in polynomial time and are of a bounded degree. In
turn, the \emph{degree} of a schemaless spanner~$S$ is the maximal
cardinality of a mapping produced over all possible documents, that
is, $\max\{ |\dom(\mu)| \mid \doc \in \Sigma^*, \mu \in S(\doc) \}$.

Formalizing the above, we can conclude from
Theorem~\ref{thm:extractioncomp} a generalization that allows for
black-box schemaless spanners. To this end, we call a representation
language $\L'$ for schemaless spanners \e{tractable} if
$\repspnrpar{\beta}(\doc)$ can be evaluated in polynomial time (for
some fixed polynomial), given $\beta\in\L'$ and $\doc \in \Sigma^*$,
and we call $\L'$ \e{degree bounded} if there is a fixed natural
number that bounds the degree of all the schemaless spanners
represented by expressions in $\L'$.

\begin{corollary}\label{cor:extractioncompext}
	Let $\L'$ be a tractable and degree-bounded representation system
	for schemaless spanners, and let $\L$ be the union of $\L'$ and the
	class of all sequential VAs.  Let $k$ be a fixed natural number and
	let $\extrtree$ be an RA tree.  The evaluation problem of
	$\extrtree$ is solvable with polynomial delay, assuming that for all
	join and difference nodes $v$ of $I[\tau]$, the left and right
	subtrees under $v$ share at most $k$ variables.
\end{corollary}

Combining such black-box schemaless spanners in the instantiated RA
tree increases the expressiveness, as it allows us to incorporate
spanners that are not (and possibly cannot be) described as RA
expressions over VAs, such as string
equalities~\cite{DBLP:journals/jacm/FaginKRV15}.  Other examples of
such spanners are part of speech (POS) taggers, dependency parsers,
sentiment analysis modules, and so on.

\newcommand{\bb}[1]{\texttt{#1}}
\begin{example}
	Following Example~\ref{ex:sec5studrec}, suppose that we now wish to
	extract the students that do not have any \e{positive}
	recommendations.  Assume we have a black-box spanner for sentiment
	analysis, namely $\bb{PosRec}$, with the variables $\exvar{stdnt}$
	and $\exvar{posrec}$, that extract names and their corresponding
	positive recommendation. Note that this spanner has the degree
	$2$. We can replace $\regex{nr}$ in the instantiation $I$ of
	Figure~\ref{fig:ratree1} with $\bb{PosRec}$, and thereby obtain the
	desired result.  If $\bb{PosRec}$ can be computed in polynomial
	time, then the resulting query can be evaluated in polynomial delay.
\end{example}

\section{Conclusions}\label{sec:conclusions}

We have studied the complexity of evaluating algebraic expressions
over schemaless spanners that are represented as sequential regex
formulas and sequential VAs. We have shown that we hit computational
hardness already in the evaluation of the natural join and difference
of two such spanners. In contrast, we have shown that we can compile
the natural join of two sequential VAs (and regex formulas) into a
single sequential VA, in polynomial time, if we assume a constant
bound on the number of common variables of the joined spanners;
hence, under this assumption, we can evaluate the natural join with
polynomial delay.  As an alternative to this assumption, we have
proposed and investigated a new normal form for sequential spanners,
namely disjunctive functional, that allows for such efficient
compilation and evaluation.  

Bounding the number of common variables  between the involved spanners
also allows to evaluate the difference with polynomial delay, even
though this cannot be obtained by compiling into a VA---an exponential
blowup in the number of states is necessary already for Boolean
spanners. Evaluation with polynomial delay is then obtained via an
ad-hoc compilation of both the spanners and the document into a VA. We
have shown how the ad-hoc approach can be used for establishing upper
bounds on general RA trees over regex formulas, VAs, and even
black-box spanners of a bounded dimension. This has been done within
the concept of \e{extraction complexity} that we have proposed as new
lens to analyzing the complexity of spanners.

We believe that our analysis has merely touched the tip of the iceberg
on the algorithms that can be devised under the guarantee of tractable
extraction complexity. In particular, we have proposed sufficient
conditions to avoid the inherent hardness of the natural join and
difference, but it is quite conceivable that less restrictive
conditions already suffice. Alternatively, are there conditions of
extractors (possibly incomparable to ours) that are both common in
practice and useful to bound the extraction complexity?

%\balance
{\bibliographystyle{plain}

\bibliography{main}}

\onecolumn

\appendix

\section{Proofs for Section~\ref{sec:join}}

\subsection{Proof of Lemma~\ref{lem:seqtosemifunc}}
\begin{replemma}{\ref{lem:seqtosemifunc}}
	\lemsectosemi
\end{replemma}

To prove this lemma, we use the following auxiliary lemma.
\begin{lemma}
	Let $A$ be a sequential $\VA$ that is semi-funcitonal for $Y\subseteq \vars(A)$ and let $x\in \vars(A)\setminus Y$. There is an algorithm that outputs a sequential VA $A^\prime$ that is equivalent to $A$ and semi-functional for $Y\cup \{x\}$ with $O(2|Q|)$ states and $O(2|\delta|)$ transitions in $O(2(|Q|+|\delta|))$ steps where $Q$ is the set of states of $A$ and $\delta$ is $A$'s transition function.   
\end{lemma}

Lemma~\ref{lem:seqtosemifunc} follows directly from the previous lemma as 
we can invoke the algorithm iteratively for all of the variables of $X$.
It is therefore left to present the proof of this lemma. The intuition is simple, we create a new VA that is equivalent to $A$ by replacing each state $q$ for which $\varconfseq_q(x) = \vc{d}$ with two states $q^{\vc{u}}$ and $q^{\vc{c}}$ and re-wiring the transition in a way that in all of the paths from $q_0$ to $q^{\vc{u}}$ variable $x$ does not appear and in all of the paths from $q_0$ to $q^{\vc{c}}$ variable $x$ was closed. 
We then prove that the resulting automaton remains semi-functional for $Y$ and is also semi-functional for $\{x\}$ (in contrast to $A$).

Since we consider two VAs $A$ and $A^\prime$, to avoid ambiguity we add to the variable configuration function $\varconfseq_q$ an upper script that indicates to  which automaton we refer to.
Formally, let us denote $A\df (Q,q_0, \delta, F)$. 
We denote by $\tilde{Q}$ the subset of $Q$ that consists of those states $\tilde{q}$ for which 
$\varconfseq^A_{\tilde{q}}(x)= \vc{d}$, these are the states that we wish to replace.
We then define that the VA 
$A^\prime \df (Q^\prime, q^{\prime}_0,\delta^{\prime}, F^{\prime})$  such that 
\begin{enumerate}
	\item
	$Q^\prime =( Q \setminus \tilde{Q}) \cup \{ \tilde{q}^{\vc{c}} \vert \tilde{q}\in \tilde{Q} \} \cup  \{ \tilde{q}^{\vc{u}} \vert \tilde{q}\in \tilde{Q} \}$ and
	\item
	if $q_0\in \tilde{Q}$ then 	$q^{\prime}_0 = q^\vc{u}_0$, otherwise 
	$q^{\prime}_0 = q_0$, and
	\item
	$F^\prime = \{q \vert q\in F \setminus \tilde{Q} \} \cup \{q^{\vc{c}}, q^{\vc{u}}\vert q\in F\cap \tilde{Q} \} $,
\end{enumerate}
and the transition function $\delta^\prime$ is as described now.
For every $(p, o, q) \in \delta$ with $\varconfseq^A_q(x) \ne \vc{d}$, we set $(p, o, q) \in \delta^\prime$. In addition, 
for every $(p, o, q) \in \delta$ with $\varconfseq^A_q(x)=\vc{d}$,
\begin{itemize}
	\item 
	if $\varconfseq^A_p(x) = \vc{d}$ then
	$(p^{\vc{u}}, o ,  q^{\vc{u}})\in  \delta^\prime $ and $(p^{\vc{c}}, o ,  q^{\vc{c}}) \in \delta^\prime$;
	\item
	if $\varconfseq^A_p(x) = \vc{o}$ then
	$(p^{\vc{o}}, o ,  q^{\vc{c}})  \in\delta^\prime $  ({ and in this case $o = \vcl{x}$});
	\item
	if $\varconfseq^A_p(x) = \vc{c}$ then
	$(p^{\vc{c}}, o ,  q^{\vc{c}}) \in \delta^\prime $.
\end{itemize}

It is left to show that 
\begin{enumerate}
	\item~\label{enumit:equiv} $A^\prime$ is equivalent to $A$ and 
	\item~\label{enumit:semif} $A^\prime$ is semi-functional for $Y\cup \{x\}$.
\end{enumerate}

To prove (\ref{enumit:equiv}) we show that every mapping in $\repspnr{A^\prime }(\doc)$ is also a mapping in $\repspnr{A}(\doc)$ and vice versa. 
Indeed, for every accepting run of $A^\prime$ there might be two case $(a)$ it includes only states from $Q$ or $(b)$ it includes also states in $Q^\prime \setminus Q$. If $(a)$ then the claim is straightforward as $\delta^\prime \cap \Big( Q \times (\Sigma \cup \{ \epsilon\} \cup \Gamma_{\vars(A)}) \times Q \Big)$ equals $\delta$. If $(b)$ then we can divide the corresponding run according to the variable configuration of $x$ and obtain that the claim is a direct consequence of the definition of $\delta^\prime$.
To show that every mapping in $\repspnr{A}(\doc)$ is also a mapping in $\repspnr{A^\prime}(\doc)$ we take the corresponding map in $A$ and divide into segments according to the variable configuration of $A$, we can then use $\delta^\prime$ definition to construct an accepting run in $A^\prime$ for the same mapping.  

To prove (\ref{enumit:semif}), we first observe that $A^\prime$ is semi-functional for $Y$ (since otherwise it would imply that $A$ is not semi-functional for $Y$). It is therefore left to show that it is semi-functional for $\{x\}$. Note that from the way we defined $\delta^\prime$ we can conclude that for every state 
$p \in \{ \tilde{q}^{\vc{u}} \vert \tilde{q} \in \tilde{Q}  \}$ it holds that 
$\varconfseq^{A^\prime}_p(x) = \vc{w}$  and  
for every state 
$p \in \{ \tilde{ q }^{\vc{c}} \vert \tilde{q} \in \tilde{Q}  \}$ it holds that $\varconfseq^{A^\prime}_p(x) = \vc{c}$. For all other states $r$ the $\varconfseq^{A^\prime}_r$ is identical to $\varconfseq^{A}_r$.

\subsection{Proof of Proposition~\ref{prop:seqtodis}}
\begin{repproposition}{\ref{prop:seqtodis}}
	\propseqtodis
\end{repproposition}

\paragraph{Sequential regex formulas to disjunctive functional regex formulas}
Let $\alpha$ be a sequential regex formula. 
We translate it into an equivalent disjunctive functional by defining the set $A(\alpha)$ of its disjuncts  recursively as follows:
\begin{itemize}
	\item
	if $\alpha = \emptyset$ then $A(\alpha) = \emptyset$;
	\item
	if $\alpha = \sigma$ then $A(\alpha) = \{\sigma\}$;
	\item
	if $\alpha = \epsilon$ then $A(\alpha) = \{\epsilon\}$;
	\item 
	if $\alpha = \alpha_1 \vee \alpha_2$ then 
	\begin{itemize}
		\item 
		if $ \vars (\alpha_1) = \vars( \alpha_2) = \emptyset$ then $A(\alpha) = \{\alpha_1 \vee \alpha_2\}$
		\item
		otherwise $A(\alpha) = A(\alpha_1) \cup A(\alpha_2)$;
	\end{itemize}
	\item 
	if $\alpha = \alpha_1 \cdot \alpha_2$ then 
	$A(\alpha) = 
	\{ \beta_1 \cdot \beta_2 \vert \beta_1 \in 	A\{ \alpha_1\}, \beta_2 \in  A\{\alpha_2 \} \}$;
	\item
	if $\alpha = (\alpha_1)^*$ then 
	$A(\alpha) = \{ \beta_0 \cdots \beta_n \vert 
	n\in \mathbb{N}, \beta_i \in A(\alpha_1)
	\}$;
	\item
	if $\alpha = x\{ \alpha_1 \}$ then $A(\alpha) = \{ x\{\beta \} \vert \beta \in A(\alpha_1) \}$
\end{itemize}
We can then conclude the desired by proving  the following lemma using a simple induction on $\alpha$'s structure.
\begin{lemma}
	If $\alpha$ is a sequential regex formula then
	$\repspnrpar{\alpha} = \bigvee_{ \gamma\in{A(\alpha)}} \repspnrpar{\gamma}$
	where each $\gamma$ is functional.
	In addition, $A(\alpha)$ is finite if $\alpha$ is finite.
\end{lemma}

\paragraph{Sequential VA to disjunctive functional VA}
Let $A$ be a sequential automaton with set of variables $\vars(A)$.
We iterate through all of the possible subsets $V$ of $\vars(A)$ and for each such subset we create a new VA $A_V$ that consists of all of the accepting runs of $A$ that include exactly the variables of $V$ and only them (this can be done, for instance, by a BFS on $A$ starting from its initial state $q_0$).
We then construct a new automaton $A^\prime$ that is the disjoint union of all of those $A_V$'s by adding a new initial state with transitions to all of the initial states of the $A_V$'s.

\section{Proofs for Section~\ref{sec:diff}}

\subsection{Proof of Lemma~\ref{thm:diffseqkvars}}
\begin{replemma}{\ref{thm:diffseqkvars}}
	\thmdiffseqkvars
\end{replemma}

First, 
if $\doc = \epsilon$ then there two possible cases:
\begin{itemize}
	\item
	$\repspnr{ A_2}(\doc) = \emptyset$ and then we set $A_{\doc}$ to be 
	$A_1$. 
	\item
	$\repspnr{ A_2}(\doc) \ne \emptyset$ and then 
	$\repspnr{ A_1 \setminus A_2}(\doc) = \emptyset$ since any two mappings are compatible. In this case we set $A_{\doc}$ to automaton for for $\emptyset$. 
\end{itemize}
Note that determining whether $\repspnr{ A_2}(\doc) = \emptyset$
can be done in polynomial time (see  Theorem~\ref{thm:enumerationseq}).

Second, we can assume that $A_2$ has no other variables except those that are in $A_1$ due to the following: 
$ \repspnrpar{A_1 \setminus A_2}(\doc) = \repspnrpar{A_1 \setminus \pi_{\vars(A_1) }A_2}(\doc) $
since 
if a mapping $\mu_1 \in \repspnrpar{A_1}(\doc)$ has a compatible mapping $\mu_2 \in \repspnrpar{A_2}(\doc)$ then  $\mu_2 \restrict  \vars(A_1)$ is also compatible for $\mu_1$ and is in  $\repspnr{\pi_{\vars(A_1)}{A_2}}(\doc)$.
On the other hand, if a mapping $\mu_1 \in \repspnrpar{A_1}(\doc)$ has a compatible mapping $\mu_2 \in \repspnrpar{\pi_{\vars(A_1)} A_2}(\doc)$, then no matter how we extend $\mu_2$, as long as we extend it with variables that are not in $\vars(A_1)$, it remains compatible to that $\mu_1$. 
Thus, if a mapping $\mu_1 \in \repspnrpar{A_1}(\doc)$ does not have a compatible mapping $\mu_2 \in \repspnrpar{A_2}(\doc)$ then it also does not have such in $\pi_{\vars(A_1)}\repspnrpar{A_2}(\doc)$.
Therefore, we may assume that $\vars(A_1) \supseteq \vars(A_2)$ and we denote $\vars(A_2)$ by $V$.

Third, we can also assume that $A_1$ is semi-functional for $V  \df  \vars(A_2)$ since if it is not we can translate it into such in polynomial time (see
Lemma~\ref{lem:seqtosemifunc}).

We now move to the construction of sequential VAs $A$ and $B$ with fixed  
$|\vars(A) \cap \vars(B)|$ for which we will prove that $\repspnr{A \join B}(\doc) = \repspnr{A_1 \setminus A_2}(\doc)$. 

\paragraph{Constructing $A$:}
For every $X\subseteq V$ we define $X^\prime = \{x^\prime \vert x \in X \}$.
We construct a VA $A$ that extends the $A_1$ such that
for every accepting run $\rho$ of $A_1$ that closes exactly the set $X$ of variables, there is a corresponding accepting run in $A$ that assigns the same values as $\rho$ to $X$, and in addition, assigns to each of the variables of $X^\prime$ the span $\mspan{1}{1}$ (indicating that the variables of $X$ were closed throughout $\rho$) and to each of the variables in $V^\prime \setminus X^\prime$ the span $\mspan{|\doc|+1}{|\doc|+1}$ (indicating that the variables of $V \setminus X$ were unseen throughout $\rho$). 

More formally, the VA $A$ consists of the disjoint copies $A_1^X$ of $A_1$ for $X\subseteq V$, along with their transitions, of an initial state $q_0$, and of an accepting state $q_f$. 
We then extend $A$ as follows:
\begin{itemize}
	\item
	for every $X$, we add a path from $q_0$ to the initial state of $A_1^X$ that opens and then immediately closes all of the variables in $X^\prime $;
	\item 
	for every $X$ and every accepting state $q$ of $A_1^X$ with $c_q(x) = \vc{c}$ for every $x\in X$, we add a path from $q$ to the final state $q_f$ that opens and then immediately closes all of the variables in $V^\prime \setminus X^\prime$.
\end{itemize}
We then remove from $A$ all of the states that are not reachable throughout any accepting run and obtain a sequential VA.

Let $\mu$ be a mapping with $\dom(\mu)\subseteq V$. A mapping $\mu^\prime$ is called the \emph{marked extension} of $\mu$ if the following hold:
\begin{itemize}
	\item $\dom(\mu^\prime) = \dom(\mu) \cup V^\prime$ 
	\item for every $x \in \dom(\mu)$ it holds that 
	$\mu^\prime(x) = \mu(x)$;
	\item
	for every $x \in \dom{(\mu)}$ it holds that 
	$\mu^\prime(x^\prime) = \mspan{1}{1}$;
	\item
	for every $x \not \in \dom{(\mu)}$ it holds that   $\mu^\prime(x^\prime) = \mspan{|\doc|+ 1}{|\doc|+1}$.
\end{itemize}
We observe that the following holds:
\begin{lemma}\label{lem:diffadhocextens}
	Let $\mu_1,\mu_2$ be two mappings with $\dom(\mu_1) =\dom(\mu_2) \subseteq V $ and let 
	$\mu^\prime_1$ be the marked extension of $\mu_1$ and  $\mu^\prime_2$ the marked  extension of $\mu_2$. The following holds: 
	$\mu_1$ is compatible with $\mu_2$ if and only if 
	$\mu^\prime_1$ is compatible with $\mu^\prime_2$ 
\end{lemma} 
\begin{proof}
	Assume that $\mu_1$ is compatible with $\mu_2$. 
	It is enough to show that 
	$\mu^\prime_1 \restrict {V^\prime} $ for and that 
	$\mu^\prime_2  \restrict {V^\prime}$ are compatible. 
	This follows directly from the fact that $\dom(\mu_1) = \dom(\mu_2)$.
	On the other hand,  assume that $\mu^\prime_1$ is compatible with $\mu^\prime_2$. 
	Since $\mu^\prime_1 \restrict V$ is identical to $\mu_1$ and $\mu^\prime_2 \restrict V$ is identical to $\mu_2$ we conclude that $\mu_1$ is compatible with $\mu_2$.
\end{proof}

The following lemmas describe the connection between $A_1$ and $A$:
\begin{lemma}\label{lem:diffseqadhocA}
	$\mu_1 \in \repspnr{A_1}(\doc)$ if and only if  $\mu \in \repspnr{A}(\doc)$ where $\mu$ is the marked extension of $\mu_1$.
\end{lemma}
\begin{proof}
	Assume that $\mu_1 \in \repspnr{A_1}(\doc)$. Then there is an accepting run on $A_1$ on $\doc$ that corresponds with $\mu_1$.
	We build an accepting run of $A$ that corresponds with the marked extension $\mu$ of $\mu_1$ in the following way.
	We denote $X= \dom(\mu_1)$.
	The run starts with a path from $q_0$ to the initial state of the copy $A_1^{X}$ that opens and then closes all of the variables in $X^\prime$. We then continue with the accepting run of $A_1^X$ on $\mu_1$ (there is such since $A_1^X$ is a copy of $A_1$). Then from the accepting state $q$ we reached we continue with a path that opens and then closes all of the variables in $V^\prime \setminus X^\prime$. We can do that from the way we constructed $A$. Hence, we conclude that $\mu \in \repspnr{A}(\doc)$.
	
	Assume that $\mu \in \repspnr{A}(\doc)$ is the marked extension of $\mu_1$. There is a run of $A$ on $\doc$ that corresponds with $\mu$. We take the sub-run on the copy of $A_1$ in $A$ and obtain an accepting run of $A$ on $\doc$ that corresponds with $\mu_1$.Thus, we can conclude that 
	$\mu_1 \in \repspnr{A_1}(\doc)$.
\end{proof}

\paragraph{Constructing $B$:}

We construct $B$ as follows:
We compute all of the assignments in $\repspnr{A_2}(\doc)$ and replace each of those with their marked extension to obtain the set $M$ of extended mappings.
We then compute all of the possible extended mappings whose domain is contained in $V \cup V^\prime$.
For every such mapping if it is not in $M$ then 
we insert it to the set $\bar{M}$.

We now define the VA $B$ as the disjoint
union of all of the paths that correspond with the assignments in $\bar{M}$ along with an initial state $q_0$ and an
accepting state $q$.  
For every path $P$ that corresponds with the assignment $\mu$ with $X\df \dom(\mu)$ we extend $B$ as follows: 
\begin{itemize}
	\item 
	we add a path that connects $q_0$ with the first state of $P$ that opens and then immediately closes all of the variables in $X^\prime$;
	\item
	we add a path that connects the last state of $P$ to the accepting state $q$, that opens and then immediately closes all of the variables in $Y^\prime$ where $Y = V \setminus X$.
\end{itemize}
Note that $B$ is a  sequential VA. 

\paragraph{Correctness:}
The following lemma that describes the connection between $A_1, A_2$, $B$ and $\doc$.
\begin{lemma}~\label{lem:diffadhocA2Bd}
	A mapping
	$\mu_1 \in \repspnr{A_1}(\doc)$ has a compatible mapping $ \repspnr{A_2}(\doc)$ 
	if and only if the marked extension $\mu$ of $\mu_1$
	does not have a compatible mapping in $ \repspnr{B}(\doc)$.
\end{lemma}
\begin{proof}
	Assume	$\mu_1 \in \repspnr{A_1}(\doc)$ has a compatible mapping $\mu_2 \in \repspnr{A_2}(\doc)$ 
	and assume by contradiction that $\mu$ has a compatible mapping $\mu^\prime_2 \in  \repspnr{B}(\doc)$.
	Then by Lemma~\ref{lem:diffadhocextens} the restriction $\mu^\prime_2\restrict V$ is compatible with $\mu_2$ which is impossible due to $B$'s definition.  	
	
	Assume that	$\mu_1 \in \repspnr{A_1}(\doc)$ does not have a compatible mapping $\mu_2 \in \repspnr{A_2}(\doc)$. By Lemma~\ref{lem:diffseqadhocA}, the extension $\mu$ of $\mu_1$ is in  $\repspnr{A}(\doc)$. Lemma~\ref{lem:diffadhocextens} and $B$'s definition implies that there is a mapping in $\repspnr{B}(\doc)$ that is compatible to $\mu$.
\end{proof}

We can conclude that $\repspnr{\pi_V (A \join
	B)}(\doc) = \repspnr{A_1 \setminus A_2}(\doc)$

Note that $\vars(A) = \vars(B) = V\cup V^\prime$ and since $|V|$ is fixed we this is also the case for $\vars(A)\cap \vars(B) = \vars(A)$. Therefore, we can use Lemma~\ref{lem:fptjoin} to conclude the desired.

\paragraph{Complexity:}
Note that throughout our construction we performed only polynomial time steps:
In case the document is empty, checking the emptiness of $\repspnr{A_2}(\doc)$ can be done in polynomial time. 
Transforming $A_1$ into a semi-functional VA for $V$ requires polynomial time assuming $|V|$ is fixed.
Constructing the extension $A$ of $A_1$ requires polynomial time since we fix $|V|$.
Computing the set of extended mappings that are not compatible with any of the extended mappings that correspond with $\repspnr{A_2}(\doc)$  also requires polynomial time since $|V|$ is fixed.

\subsection{Proof of Theorem~\ref{thm:diffwonehard}}
\begin{reptheorem}{\ref{thm:diffwonehard}}
	\thmdiffwonehard
\end{reptheorem}

Let $\varphi=C_0\wedge\cdots\wedge C_m$ be a 3CNF fomula with variables $x_1,\ldots,x_n$ and
denote by $l_{i,1},l_{i,2},l_{i,3}$ the literals of the $i$-th clause for $0\leq i\leq m$. 
To show the reduction, we use the same idea of the reduction in the proof of Theorem~\ref{thm:joinseqhard}.

Set $\doc=s_1\ldots s_n$, where every $s_i$ is unique (this can be achieved by two distinct elements in
$\alphabet$ such that each $|s_i|\in O(\log(n))$) and let $S=\{s_1,\ldots,s_n\}$.
Further set $\alpha_S=\bigvee_{\sigma\in S}\sigma$. We define the regex formula $\alpha_1$, using only
the $k$ variables $V=\{y_1,\ldots,y_k\}$ as
\begin{equation*}
\alpha_1=\alpha_S^*\bind{y_1}{\alpha_S}\alpha_S^*\bind{y_2}{\alpha_S}\ldots\bind{y_k}{\alpha_S}\alpha_S^*.
\end{equation*}
We next define for every $0\leq i\leq m$
the regex $\alpha_{C_i}$, such that $\repspnr{\alpha_{C_i}}(\doc)$ corresponds
to all possible assignments of $\varphi$ of weight $k$, such that $C_i$ is not satisfied. Let $S_i\subseteq S$
such that such that $s_j\in S_i$ iff $x_j$ is contained in clause $C_i$. Further let $S_i^-$ denote the variables
in $S_i$ which appear negated in $C_i$, i.e. $s_j\in S_i^-$ iff $x_j$ is a negated variable in $C_i$, and
define $S_i^+=S_i\setminus S_i^-$. Let $\mathit{ind}(S_i^-):=\{1\leq j\leq n\mid s_i\in S_i^-\}$
and similarly $\mathit{ind}(S_i^+):=\{1\leq j\leq n\mid s_i\in S_i^+\}$.
Further define the regular expression
$\alpha_{S\setminus S_i^+}:=\bigvee_{\sigma\in S\setminus S_i^+}\sigma$.
In order to define $\alpha_{C_i}$, we need to consider different cases on the number of positive
and negative literals in $C_i$.
\begin{itemize}
	\item  First, assume that
	$|S_i^+|=3$, i.e. $C_i$ is a clause containing three positive variables. Then we set 
	\begin{equation*}
	\alpha_{C_i}=\alpha_{S}^*\bind{y_1}{\alpha_{S\setminus S_i^+}}\alpha_{S}^*
	\ldots\bind{y_k}{\alpha_{S\setminus S_i^+}}\alpha_{S}^*.
	\end{equation*}
	
	\item Next assume that $|S_i^+|=2$ and thus there is some $1\leq j\leq n$ such that $\{j\}=\mathsf{ind}(S_i^-)$.
	For every $1\leq u\leq k$, we define
	\begin{align*}
	\alpha_{C_i}^u=&\alpha_{S}^*\bind{y_1}{\alpha_{S\setminus S_i^+}}\alpha_{S}^*
	\ldots \alpha_{S}^*\bind{y_{u-1}}{\alpha_{S\setminus S_i^+}}\\
	&\cdot\alpha_{S}^*\bind{y_u}{s_j}\alpha_{S}^*
	\bind{y_{u+1}}{\alpha_{S\setminus S_i^+}}\\
	&\cdot\alpha_{S}^*\ldots
	\alpha_{S}^*\bind{y_k}{\alpha_{S\setminus S_i^+}}\alpha_{S}^*.
	\end{align*}
	Note that the regex $\alpha_{C_i}^u$ is built similar to $\alpha_{C_i}$ in the first case,
	except for the part $\bind{y_u}{s_j}$ instead of $\bind{y_u}{S\setminus S_i^+}$.
	We set $\alpha_{C_i}=\bigvee_{1\leq u\leq k} \alpha_{C_i}^u$.
	
	\item  The case where $|S_i^+|=1$
	is defined similar to the case where $|S_i^+|=2$. Assume that there are $1\leq j_1< j_2\leq n$
	such that $\{j_1,j_2\}=\mathsf{ind}(S_i^-)$. Then for every $1\leq u_1<u_2\leq k$ we define
	\begin{align*}
	\alpha_{C_i}^{u_1,u_2}=
	&\alpha_{S}^*\bind{y_1}{\alpha_{S\setminus S_i^+}}\alpha_{S}^*
	\ldots \alpha_{S}^*\bind{y_{u_1-1}}{\alpha_{S\setminus S_i^+}}\\
	&\cdot\alpha_{S}^*\bind{y_{u_1}}{s_{j_1}}\alpha_{S}^*
	\ldots\alpha_{S}^*\bind{y_{u_2-1}}{\alpha_{S\setminus S_i^+}}\\
	&\cdot\alpha_{S}^*\bind{y_{u_2}}{s_{j_2}}\alpha_{S}^*
	\bind{y_{u_2+1}}{\alpha_{S\setminus S_i^+}}\\
	&\cdot\alpha_{S}^*\ldots\alpha_{S}^*\bind{y_k}{\alpha_{S\setminus S_i^+}}\alpha_{S}^*.
	\end{align*}
	and set $\alpha_{C_i}=\bigvee_{1\leq u_1<u_2\leq k} \alpha_{C_i}^u$.
	
	\item 
	If $|S_i^+|=0$, then 
	$\alpha_{C_i}$ is defined analogously to the case $|S_i^+|=1$ but with three indices
	$u_1,u_2,u_3$ instead of $u_1,u_2$. 
\end{itemize}

To define $\alpha_2$, we set $\alpha_2=\bigvee_{1\leq i\leq m}\alpha_{C_i}$. Note that 
$|\alpha_2|\leq (m + 1)\cdot n^3\cdot (m + 1)^3$. It remains to show that $\varphi$ is satisfyable (with weight $k$) iff
$\repspnr{\alpha_1-\alpha_2}(\doc)\neq\emptyset$. 

So let $\tau$ be a satisfying truth assignment of weight $k$.
We claim that there exists some
$\mu\in\repspnr{\alpha_1-\alpha_2}(\doc)$, with $\mu(y_i)=\spn{j,j+1}$ for every
$1\leq i\leq k$ and $1\leq j\leq n$ such that $\tau(x_j)=\true$, and $x_j$ is the $i$-th  such variable (i.e.
among the variables $x_1,\ldots,x_{j-1}$, $i-1$ are set to true via $\tau$).
First consider the set
$\repspnr{\alpha_1}(\doc)$. It is easy to see that $\mu'\in\repspnr{\alpha_1}(\doc)$ iff
there are $1\leq j_1<\ldots<j_k\leq n$ with
$\mu'=\{y_1\mapsto \spn{j_1,j_1+1}, \ldots, y_k\mapsto \spn{j_k,j_k+1}\}$,
hence $\mu\in \repspnr{\alpha_1}(\doc)$. To see 
that $\mu\not\in \repspnr{\alpha_2}(\doc)$, 
we need to consider the set $\repspnr{\alpha_2}(\doc)$.
By definition of $\alpha_2$, 
we have that $\mu'\in\repspnr{\alpha_2}(\doc)$ iff 
there is some $1\leq i\leq m$ 
such that $\mu'\in\repspnr{\alpha_{C_i}}(\doc)$.
The containment
$\mu'\in\repspnr{\alpha_{C_i}}(\doc)$ can be characterized as follows:
There exist $j_1,\ldots,j_k\in\{1,\ldots,n\}\setminus\mathsf{ind}(S_i^+)$
such that $1\leq j_1<\ldots<j_k\leq n$,
$\mathsf{ind}(S_i^-)\subset\{j_1,\ldots,j_k\}$  and
$\mu'=\{y_1\mapsto \spn{j_1,j_1+1},\ldots, y_k\mapsto \spn{j_k,j_k+1}\}$.

For some $1\leq i\leq n$, first assume that $C_i$ is a clause only containing positive variables.
Since $\tau$ is a satisfying assignment, there is some $i'$ such that 
and $\tau(x_{i'})=1$. Hence $\mu(y_{i'})=\spn{j,j+1}$ for some $1\leq j\leq n$, and since
$i'\in\mathsf{S_i^+}$ we have that $\mu\not\in \repspnr{\alpha_2}(\doc)$.

Fix some $1\leq i\leq n$. Since $\tau$ is a satisfying assignment, $\tau(x_{i'})=1$
for some $i'\in\mathsf{ind}(S_i^+)$ or $\tau(x_{i'})=0$ for some $i'\in\mathsf{ind}(S_i^-)$.
In the first case, $j_1,\ldots,j_k\in\{1,\ldots,n\}\setminus\mathsf{ind}(S_i^+)$ does not hold
, and in the second case $\mathsf{ind}(S_i^-)\subset\{j_1,\ldots,j_k\}$ is violated.
Thus $\mu\not\in \repspnr{\alpha_{C_i}}(\doc)$ 
for every $1\leq i\leq n$, and hence
$\mu\not\in \repspnr{\alpha_2}(\doc)$.

For the other direction, let $\mu\in\repspnr{\alpha_1-\alpha_2}(\doc)$. Using the same arguments as
above, we can show that the truth assignment $\tau$ is a satisfying truth assignment of $\varphi$ with
weight $k$, where $\tau(x_i)=1$ iff $y_i=\spn{j,j+1}$ for some $1\leq j\leq n$.

\subsection{Proof of Lemma~\ref{lem:syncregextosyncautomaton}}
\begin{replemma}{\ref{lem:syncregextosyncautomaton}}
	\lemsyncregextosyncautomaton
\end{replemma}
\begin{proof}
	Like in Lemma~3.4 in~\cite{DBLP:conf/pods/FreydenbergerKP18}, we obtain~$A$ using the Thompson construction 
	(cf.~e.g.,~\cite{hop:int}) for converting a regular expression into an $\epsilon$-NFA, where we treat variable operations like symbols. More specifically, each occurrence of a variable binding $\bind{x}{\gamma}$ is interpreted like $\vop{x}\cdot \gamma\cdot\vcl{x}$. 
	
	As $\gamma$ is sequential, $A$ is also sequential. A feature of the Thompson construction is that for every occurrence of a symbol the regular expression is converted into an initial state and a final state, where the former has a transition to the latter that is labeled with that symbol. All new transitions that enter to this sub-automaton pass through the initial state; the final state can only be reached through the original transition. This creates a new target state for every variable operation.
	
	The other condition for VAs that are synchronizing for $X$ follows directly from the fact that $\gamma$ has no disjunctions over $x$. Either $\gamma$ never  uses $x$ (e.\,g., if $\gamma=\bind{x}{\emptyset}$, if $\gamma=\bind{x}{a}\cdot\emptyset$, or due to other uses of $\emptyset$); then the accepting runs of $A$ also never operate on $x$. Or $\gamma$ always uses $x$; then the same also holds for all accepting runs of $A$.
\end{proof}

\subsection{Proof of Proposition~\ref{prop:syncnotexpressive}}
\begin{repproposition}{\ref{prop:syncnotexpressive}}
	\propsyncnotexpressive
\end{repproposition}
\begin{proof}
	Note  that $\repspnrpar{\gamma}(\mathsf{a}\mathsf{a} ) \ne \emptyset$ and that also $\repspnrpar{\gamma}(\mathsf{b}\mathsf{b} ) \ne \emptyset$; however, $\repspnrpar{\gamma}(\mathsf{a}\mathsf{b} ) = \emptyset$.
	
	Note also that $\repspnrpar{\gamma}(\mathsf{a}\mathsf{a} ) =\{ \mu_1 \}$ where $\mu_1$ is the mapping that maps $x$ to $\mspan{2}{3}$ and that  $\repspnrpar{\gamma}(\mathsf{b}\mathsf{b} ) =\{ \mu_2 \}$ where $\mu_2$ is the mapping that maps $x$ to $\mspan{2}{3}$.
	
	We can conclude that there exists states $q_1,q_2,q_3 \in Q$ and $q_4 \in F$ such that that the following is a valid an accepting run of $A$ (on $\mathsf{aa}$):
	$$  \rho_1 \df (q_0,1) 
	\overset{\mathsf{a}}{\rightarrow}  
	(q_1,2)
	\overset{\vop{x}}{\rightarrow} 
	(q_2,2)
	\overset{\vcl{x}}{\rightarrow}
	(q_3,2)  
	\overset{\mathsf{a}}{\rightarrow}
	(q_4,3)
	$$
	Similarly, there exists states  $p_1,p_2,p_3 \in Q$ and $p_4 \in F$ such that that the following is a valid an accepting run of $A$ (on $\mathsf{bb}$):
	$$ \rho_2 \df (q_0,1) 
	\overset{\mathsf{b}}{\rightarrow}  
	(p_1,2)
	\overset{\vop{x}}{\rightarrow} 
	(p_2,2)
	\overset{\vcl{x}}{\rightarrow}
	(p_3,2)  
	\overset{\mathsf{b}}{\rightarrow}
	(p_4,3)
	$$
	Let us assume by contradiction that $A$ is synchronized for the variable $x$.
	Thus, we obtain that $p_2 = q_2$.
	So we can take the first run of $\rho_1$ and the second of $\rho_2$ and glue them together to obtain a new run:
	$$  \rho_{1,2} 
	\df 
	(q_0,1) 
	\overset{\mathsf{a}}{\rightarrow}  
	(q_1,2)
	\overset{\vop{x}}{\rightarrow} 
	(q_2,2)
	\overset{\vcl{x}}{\rightarrow}
	(p_3,2)  
	\overset{\mathsf{b}}{\rightarrow}
	(p_4,3)
	$$
	This run is valid and accepting and therefore we can conclude that $\repspnrpar{A}(\mathsf{ab})$ is not empty (as it contains a mapping that maps $x$ to the span $\mspan{2}{3}$).
	However, we already noted that $\repspnrpar{\gamma}(\mathsf{a}\mathsf{b} ) = \emptyset$ and we assumed that $A$ and $\gamma$ are equivalent, which lead us to the desired contradiction. 
\end{proof}

\subsection{Proof of Theorem~\ref{thm:thmdiffseqsync}}
\begin{reptheorem}{\ref{thm:thmdiffseqsync}}
	\thmdiffseqsync
\end{reptheorem}
\newcommand{\trap}{\mathsf{trap}}
\newcommand{\veclos}{\mathcal{VE}}
For two sequential VAs $A_1$ and $A_2$, let $X\df\vars(A_1) \cap \vars(A_2)$, and assume  that $A_1$ is semi-functional for $X$ and that  $A_2$ is synchronized for $X$. Before we consider the document, we fix some assumptions on $A_2$ that are possible without loss of generality.

\subsubsection*{Assumptions on $A_2$:}
First of all, we assume that $\vars(A_2)=X$. Variables in $\vars(A_2)\setminus\vars(A_1)$ play no rule for the difference $A_1 \setminus A_2$, so this assumption does not change the correctness of the construction.  We can guarantee this by replacing all transitions with operations for variables from $\vars(A_2)\setminus \vars(A_1)$ with $\epsilon$-transitions.  

As pointed out in Lemma~3.8 of~\cite{DBLP:conf/pods/FreydenbergerKP18}, this change turns $A_2$ into an automaton for $\pi_X(A_2)$. Although~\cite{DBLP:conf/pods/FreydenbergerKP18} only remarks this for functional automata, the same argument translates to sequential automata; and the translation can be performed in time that is linear in the size of the transition relation of $A_2$. 

As $A_2$ is synchronized for $X$, for every $x\in X$, either all accepting runs of $A_2$ operate on $x$, or no accepting run of $A$ operates on $x$. We assume without loss of generality that the former is the case (this can be achieved by trimming $A_2$ by removing all states that are not reachable from the initial state, and from which no final state can be reached).
As every variable operation from $\xalphabet_X$ has a unique target state, we know that $A_2$ is semi-functional for $X$. But as we assume that $\vars(A_2)=X$, and we have established that every accepting run of $A_2$ acts on all variables of $X$, this implies that $A_2$ is in fact functional.

\subsubsection*{Match structures in general:}
We are now ready for the main part of the proof, which uses constructions from the proof of Theorem~3.3 in~\cite{DBLP:conf/pods/FreydenbergerKP18}, which describes a polynomial delay algorithm for $\repspnr{A}(\doc)$ for functional VAs $A$ and documents $\doc$. We shall sketch all details that are relevant to the present proof; more information can be found in Section~4 of that paper, and also its preprint at~\url{http://arxiv.org/abs/1703.10350}.

Given a document $\doc= \sigma_1\cdots \sigma_\ell$ with $\ell\geq 1$, the key idea of the construction is to represent each element of $\repspnr{A}(\doc)$ as a sequence $c_0,\dots,c_{\ell+1}$ of variable configurations of $A$. Each $c_i$ is the last variable configuration before $\sigma_i$ is processed; and $c_{\ell+1}$ is the configuration where all variables have been closed.

To obtain these configurations, we construct the \emph{match graph} of $A$ on $\doc$, a directed acyclic graph $G(A,\doc)$ that has one designated source node, and nodes of the form $(i,q)$, where $0\leq i \leq \ell$ and $q$ is a state of $A$. Intuitively, the node $(i,q)$ represents that after processing $\sigma_1\dots\sigma_i$ (the first $i$ letters of $\doc$), $A$ can be in state $q$. 

Consequently, an edge from $(i,p)$ to $(i+1,q)$ represents that if $A$ is in state $p$ and reads the symbol $\sigma_{i+1}$, it can enter state $q$ (not necessarily directly; it may process arbitrarily many variable operations or $\epsilon$-transitions after reading $\sigma_{i+1})$. The match graph $G(A,\doc)$ can be constructed from $A$ and $\doc$ directly from the transition relation of $A$ and by using standard reachability algorithms. These reachability algorithms can also be used to trim the match graph: We remove all nodes that cannot be reached from the source, and all nodes that cannot reach a node $(\ell+1,q)$ where $q$ is a final state of $A$. As $A$ is functional, each state $q$ has a well-defined variable configuration $\varconf_q$, which means that we can also associate each node $(i,q)$ with that variable configuration.

We can now interpret the match graph $G$ as an NFA $M(A,\doc)$ over the alphabet of variable configurations in the following way: All nodes of $G(A,\doc)$ become states of $M(A,\doc)$. The source of $G(A,\doc)$ becomes the initial state, and the final state of $M(A,\doc)$ are those nodes $(\ell+1,q)$ where $q$ is a final state of $A$. Finally, each edge from a node $v$ to a node $(i,q)$ in $G(A,\doc)$ becomes a transition with the letter $\varconf_q$. We call the automaton $M(A,\doc)$ the \emph{match structure} of $A$ on $\doc$.

As shown in~\cite{DBLP:conf/pods/FreydenbergerKP18}, there is a one-to-one correspondence between the elements of $\repspnr{A}(\doc)$ and the words in the language of $A_G$.

\subsubsection*{Preliminaries to determinizing $M(A_2,\doc)$:}
In the present proof, we start with the same construction, and first construct the match structure $M(A_2,\doc)$ of $A_2$ on $\doc$, where $\doc= \sigma_1\cdots \sigma_\ell$ with $\ell\geq 1$ is our specific input document.

We have already established that we can assume that $A_2$ is functional; hence, we can directly use the construction that we discussed previously.
Our next goal is to determinize $M(A_2,\doc)$, and use this to implicitly construct a VA for the complement of $\repspnr{A_2}(\doc)$. This can then be combined with $A_1$ using a minor variation of the join-construction from Lemma~\ref{lem:joinseqsemi}.

In general, determinizing a match structure is not a viable approach (this was already observed in~\cite{DBLP:conf/pods/FreydenbergerKP18}). But in our specific case, we can use that $A_2$ is synchronizing for $X$ (which we assume to be identical with $\vars(A_2)$. As every variable operation from $\xalphabet_X$ has a unique target state (and as $A_2$ is functional), we know that every accepting run of $A_2$ executes the operations in the same order. Hence, we can define $\omega_1,\ldots,\omega_{2k}\in \xalphabet_X$ with $k\df |X|$ according to this order (i.\,e., the $i$-th operation in every accepting run is $\omega_i$). Accordingly, we define a sequence $c_0,c_1\ldots,c_{2k}$ of variable configurations, where $c_0$ is the configuration that assigns $\vc{w}$ to all variables of $x$, and $c_{i+1}$ is the configuration that is obtained from applying operation $\omega_i$ to configuration $c_i$. Note that, as $A_2$ is functional, we use variable configurations (that use $\vc{w}$) instead of extended variable configurations (that use $\vc{u}$). Later on, this will help us distinguish between the configurations of the match structure of $A_2$ and the extended configurations of $A_1$.

As the order of the possible variable operations of $A_2$ is fixed, the language of the match structure of $A_2$ on any document (not just our specific document $\doc$) is a subset of $c_0^* \cdot c_1^* \cdots c_{2k}^*$. Note that this does not mean that every configuration would appear; if two variable operations are always performed without consuming a symbol between them, the intermediate variable configuration would never appear. Consider the following functional VA that is synchronized for all it variables:
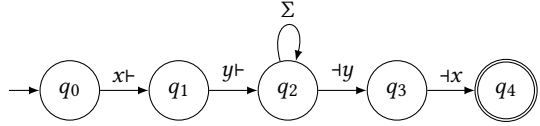
\begin{center}
	\begin{tikzpicture}[on grid, node distance =1.45cm,every loop/.style={shorten >=0pt}]
	\node[state,initial text=,initial by arrow] (q0) {$q_0$};
	\node[state,right= of q0] (q1) {$q_1$};
	\node[state,right=of q1] (q2) {$q_2$};
	\node[state,right= of q2] (q3) {$q_3$};
	\node[state,right= of q3,accepting] (q4) {$q_4$};
	\path[->]
	(q0) edge node[above] {$\vop{x}$} (q1)
	(q2) edge[loop above] node {$\Sigma$} (q2)
	(q3) edge node[above] {$\vcl{x}$} (q4)
	(q1) edge node[above] {$\vop{y}$} (q2)
	(q2) edge node[above] {$\vcl{y}$} (q3)
	;
	\end{tikzpicture}
\end{center} 
Then on every document $\doc$, its match structure accepts only a single sequence of variable configurations, namely $\varconf_{q_0}\cdot \varconf_{q_2}^{|\doc|}\cdot \varconf_{q_3}$. If we replace the state $q_2$ with an arbitrarily complicated NFA, the same observation would hold for all documents that belong to the language of the NFA (all others would be rejected).

The fact that $M(A_2,\doc)$ expresses a concatenation of unary languages is not enough to allow for determinization; for this, we need to use the unique target states. Assume that $M(A_2,\doc)$ reads a new variable configuration. More formally and more specifically, assume that $G(A_2,\doc)$ contains edges 
\begin{itemize}
	\item from some $(i,p)$ to some $(i+1,q)$ with $\varconf_p\neq \varconf_q$, and
	\item from some $(i,p')$ to some $(i'+1,q')$ with $\varconf_{p'}\neq \varconf_{q'}$, 
\end{itemize}
such that $\varconf_q=\varconf_{q'}$. Choose $j$ such that $c_j=\varconf_q$. Then the transition from $p$ to $q$ and the transition from $p'$ to $q'$ both execute $\omega_j$, which means that they pass through its unique target state. Both can then execute different sequences $\emptyword$-transitions, which means that $q\neq q'$ may hold. But we can switch these sequences of transitions, and observe that $G(A_2,\doc)$ must also contain an edge  from  $(i,p)$ to  $(i+1,q')$, and an edge from  $(i,p')$ to  $(i'+1,q)$.

Thus, for every variable configuration $j$, we can define a set $I_j$ that contains exactly those $q$ such that $M(A_2,\doc)$ contains an edge from some node $(i,p)$ to some node $(i+1,q)$ with $\varconf_p\neq \varconf_q$. In other words, $I_j$ contains those states that of $A_2$ that are encoded in states that $M(A_2,\doc)$ can enter when first reading variable configuration $c_j$. And as we established in the previous paragraph, $A_2$ being synchronized for all its variables ensures that each $I_j$ is well-defined. We also use $Q_q$ to refer to that $I_j$ for which $c_j=\varconf_q$ holds.

\subsubsection*{Determinizing $M(A_2,\doc)$:}
We use this to these insights to turn the NFA $M(A_2,\doc)$ into a DFA $D_2$ over the alphabet $c_0,\dots,c_{2k}$. Apart from the special initial state, each state of $D_2$ is a triple $(i,s,Q)$, where 
\begin{itemize}
	\item $i$ has the same role as in $G(A_2,\doc)$, meaning that it encodes how many positions of $\doc$ have been consumed,
	\item $s\leq i$ denotes when the current variable configuration was consumed the first time,
	\item $Q$ is a set of states of $A_2$, and all its elements have the same variable configuration.  
\end{itemize}
We construct $D_2$ by performing a variant of the power set construction that, in addition to taking the special structure of $M(A_2,\doc)$ into account, also includes a reachability analysis. 

The first set of states and transitions is computed as follows: For every variable configuration $c_j$ with $0\leq j\leq 2k$, we compute a set $Q_j$ that contains exactly those states $q$ such that $\varconf_q=c_j$ and  $M(A_2,\doc)$ contains a state $(0,q)$. If $Q_j\neq\emptyset$, we extend $D_2$ with the state $(0,0,Q_j)$  and a transition with label $c_j$ from the initial state to $(0,0,Q_j)$. In this case, $Q_j=I_j$ holds.

To compute the successor states of some state $(i,s,P)$, we proceed as follows: First, we consider successors with the same variable configuration. Let $c_j$ be the variable configuration of all states in $P$. We compute the set $Q_j$ of all states $q$ such that $M(A_2,\doc)$ contains a transition with label $c_j$ from some state $(i,p)$ to $(i+1,q)$. If $Q_j\neq\emptyset$, we extend $D_2$ with the state $(i+1,s,Q_j)$ and a transition with label $c_j$ from the $(i,s,P)$ to $(i+1,s,Q_j)$. 

To determine successors with different configurations, we consider all $j'$ with $j<j'\leq 2k$. For every such $c_j$, we check if $M(A_2,\doc)$ contains a transition from some $(i,p)$ with $p\in P$ to some $(i+1,q)$ with $\varconf_q=c_{j'}$. If this is the case, we extend $D_2$ with an edge from the state $(i,s,P)$ to the  state $(i+1,i+1,I_{j'})$ that is labeled with $c_{j'}$, and add the state $(i+1,i+1,I_{j'})$ to $D_2$ if it does not exist yet.

Intuitively, $D_2$ simulates $M(A_2,\doc)$ by disassembling it into sub-automata for the unary languages of each $c_j$. It needs to keep track of where in $\doc$ the current part started (using the middle component of each state triple) to avoid accepting words of the wrong length. Note that for each state $(i,s,Q)$ with $s=0$, we have $Q=I_j$ for some variable configuration $c_j$ with $0\leq j \leq 2k$. Hence, and by the definition of $D_2$, we observe that each component $Q$ is determined by the combination of $i$, $s$, and $c_j$. Therefore, we can bound the number of states in $D_2$ by $O(\ell^2 k)$, as $1\leq s\leq i \leq \ell+1$ and $0\leq j \leq 2k$ hold. Likewise, the number of transitions is bounded by $O(\ell^2 k^2)$. 

\subsubsection*{Combining $A_1$ and $D_2$:}
The actual construction of $A_{\doc}$ is a variant of the proof of Lemma~\ref{lem:joinseqsemi}: To fix identifiers,  we declare that  $A_1=(V_1,Q_1,q_{0,1},F_1,\delta_1)$ and $D_2=(Q_2,q_{0,2},F_2,\delta_2)$.  As $A_2$ is sequential, every state $q_i\in Q_1$ has an extended variable configuration $\varconfseq_q$; and as $A_1$ is semi-functional for $X$, we know that $\varconfseq_q(x)\in\{\vc{u},\vc{o},\vc{c}\}$ for all $x\in X$. Although $D_2$ is technically a DFA over the alphabet of variable configurations $c_0,\ldots,c_{2k}$, its definition allows us to associate each state in $Q_2$ with a variable configuration $\varconf_q$ by considering the incoming transitions of $q$ (which is equivalent to considering the variable configuration of the states that are encoded in $q$), and by setting $\varconf_{q_{0,2}}(x)=\vc{w}$ for all $x\in X$.

Recall that our goal is to construct a sequential VA $A$ with  $\repspnrpar{A_{\doc}}(\doc) = \repspnrpar{A_1 \setminus A_2} (\doc)$. 
In principle, $A_{\doc}$ simulates $A_1$ and $A_2$ in parallel on the input document $\doc$. But instead of using $A_2$, we use its deterministic representation $D_2$. 
If  $A_1$ picks a next state $q_1$, its opponent $D_2$ tries to counter that by picking an a state that has consistent behavior on the variables of $X$ (as $D_2$ has variable configurations as input, it hides the states of $A_2$ behind a layer of abstraction).
If $D_2$ can pick such a state of its states that is consistent, it can follow $A_1$ for the current input. If all available states are inconsistent, it cannot follow $A_2$ and has to enter the trap state. This means that the finishing the current run of $A_1$ will lead to an element of $ \repspnrpar{A_1 \setminus A_2} (\doc)$.

One technical problem is that $A_1$ may choose to leave variables from $X$ undefined; and by definition,  $\repspnrpar{A_1 \setminus A_2} (\doc)$ only takes the common variables into account. Thus, if $D_2$ is forced at some point to open a variable that was not yet processed, it cannot decide whether $A_1$ will open that variable later in the run; and using the finite control to keep track of the arising combinations would lead to an exponential blowup.

But we can work around this problem: As $A_1$ is semi-functional for $X$, each final state contains information which variables were skipped in every run that ends in that state.
For each $q\in F_1$ and each variable $x\in X$, we say that $q$ \emph{skips} $x$ if $\varconfseq_q(x)=\vc{u}$. Let $S(q)$ denote the set of skipped variables in $q$. The idea is to decompose $A_1$ into sub-automata that all skip the same variables. There are at most $|F_1|$ different sets $S(q)$ with $q\in F_1$. We call these $S_1,\ldots,S_f$ with $f\leq |F_1|$. 

For each $S_j$, we create a sub-automaton $A_{i,j}=(V_1,Q_{1,j},q_{0,1},F_{1,j},\delta_{i,j})$ that accepts exactly those runs of $A_1$ that end in a state $q$ with $S(q)=S_j$. We first mark $q_{0,1}$ and all $q\in F_1$ with $S(q)=S_f$. Then we use a standard reachability algorithm to mark all states and transitions that lead from $q_{0,1}$ to some marked finite state. Finally, we obtain $A_{i,j}$ by removing all unmarked states and transitions. The resulting automaton is sequential and semi-functional for $X$, and all accepting runs skip the same variables.
Furthermore, we observe that $\repspnrpar{A_1}=\bigcup_{j=1}^f\repspnrpar{A_{i,j}}$. 

The last major step is creating automata  $A_{\doc,j}$ such that $\repspnrpar{A_{\doc,j}}(\doc) = \repspnrpar{A_{1,j} \setminus A_2} (\doc)$. Recall that each $A_{1,j}$ skips exactly the variables from $S_j$. 

For this, we define the notions of consistent and inconsistent state pairs.
For each $q_1\in Q_{i,j}$, each $q_2\in Q_2$, and each $x\in X$, we say $q_1$ and $q_2$ are \emph{consistent} for $x$ if one of the following conditions holds:
\begin{itemize}
	\item $x\in S_j$, 
	\item $\varconfseq_{q_1}(x)=\vc{u}$ and  $\varconf_{q_2}(x)=\vc{w}$, 
	\item $\varconfseq_{q_1}(x)=\varconf_{q_2}(x)\in \{\vc{o}, \vc{c}\}$.
\end{itemize}
Otherwise, $q_1$ and $q_2$ are \emph{inconsistent}.
Building on these definitions, we say that 
\begin{itemize}
	\item $q_1$ and $q_2$ are \emph{consistent} if they are consistent for all $x\in X$,
	\item $q_1$ and $q_2$ are \emph{inconsistent} if they are inconsistent for at least one $x\in X$.
\end{itemize}
The set of states of $A_{\doc,j}$ shall consist of states of the following type:
\begin{itemize}
	\item consistent pairs $(q_1,q_2)$, with $q_1\in Q_{1,j}$ and $q_2\in Q_2$,
	\item pairs from $Q_{1,j}\times \{\trap\}$,
	\item a number of unnamed helper states.
\end{itemize}
Here $\trap$ is a special trap state that we shall use to denote that something happened that made the parallel simulations of $A_{1,j}$ and $D_2$ inconsistent. 

We define the initial state of $A_{\doc,j}$ as $(q_{0,1},q_{0,2})$ and its set of final states as $F_1\times\{\trap\}$. This matches the intuition that $\trap$ denotes that the two automata have inconsistent behavior on the variable operations. Following the same intuition, we define that for every transition from some state $p$ to some state $q$ in $A_{1,j}$, we have a transition with the same label from $(p,\trap)$ to $(q,\trap)$ in $A_{\doc,j}$. 

To define the ``main behavior'' of $A_{\doc,j}$, we use the notion of a the \emph{variable-$\emptyword$-closure} as in Lemma~\ref{lem:joinseqsemi}: For every $p\in Q_{1,j}$, we define 
$\veclos(p)$ as the that of all $q\in Q_{1,j}$ that can be reached from 
$p$ by using only transitions from   $\{\emptyword\}\cup\xalphabet_{V_i}$. 

Let $\hat{Q}_2\subset Q_2$ be the set of all states to which $q_{0,2}$ has a transition. 
For each $q_1\in \veclos(q_{0,1})$, we distinguish the following cases:
\begin{itemize}
	\item If there is some $q_2\in \hat{Q}$ such that $q_1$ and $q_2$ are consistent, we add a state $(q_1,q_2)$ to~$A_{\doc}$. 
	\item If  all $q_2\in\hat{Q}$ are inconsistent with $q_1$, we add a state $(q_1,\trap)$ to~$A_{\doc}$.
\end{itemize} 
In both cases, we connect $q_{0,1} $ with  the new state using a sequence of helper states that has exactly the same variable operations and $\emptyword$-transitions as one that takes $Q_1$ from $q_{0,1}$ to $q_1$. We consider all states $(q_1,q_2)$ and  $(q_1,P)$ that were introduced in this step states on level 0. Intuitively, state of the form $(q_1,q_2)$ describe cases where $D_2$ can follow the behavior of $A_{1,j}$, states of the form 
$(q_1,\trap)$ describe cases where it cannot follow the behavior and has to give up. 

Now, we successively process the symbols $\sigma_i$ of $\doc=\sigma_1\cdots\sigma_{\ell}$. For each $i$ with $1\leq i\leq \ell$, we process the states on level $i-1$ and compute their successors as follows: 

For each state $(p_1,p_2)$ on level $i-1$, we define $\hat{Q}\subset Q_2$ as the set of all states to which $p_{2}$ has a transition. We then consider each $q_1\in\bigcup_{(p_1,\sigma_i,q')\in \delta_{1,j}}\veclos(q')$ and distinguish the same cases as for level $0$. Now, the new non-trap states are on level $i$; and in each case, $(p_1,p_2)$ is connected with the new state using a sequence of helper states that processes $\sigma_{i}$ and then exactly the same variable operations and $\emptyword$-transitions that take $A_1$ from $q'$ to $q_1$.

Now observe that every run of $A_{1,j}$ maps into a run of $A_{\doc,j}$ (and, likewise, every run of $A_1$ that skips exactly the variables of $S_j$ maps into a run of $A_{i,j}$). If this run can be matched by any run of $A_2$ (as realized by $D_2$), it ends in a state $(q_1,q_2)$ and is not an accepting run of $A_{\doc,j}$. But if it is not matched, then $D_2$ will be forced to admit the inconsistency, and redirect the simulation into the copy of $A_{1,j}$ that assigns the $\trap$ state to $D_2$. In other words, $\repspnrpar{A_{\doc,j}}(\doc) = \repspnrpar{A_{1,j} \setminus A_2} (\doc)$.

Finally, we obtain $A_{\doc}$ by taking all $A_{\doc,j}$ and adding a new initial state that has an $\emptyword$-transition to each initial state of some $A_{\doc,j}$. Then 
$$\repspnrpar{A_{\doc}}(\doc) = \bigcup_{j=1}^f\repspnrpar{A_{\doc,j}}(\doc) =  \bigcup_{j=1}^f\repspnrpar{A_{1,j}\setminus A_2}(\doc) =  \repspnrpar{A_{1}\setminus A_2}(\doc).$$
As $A_{1}$ is sequential, every $A_{1,j}$ is sequential. Therefore, every $A_{\doc,j}$ is sequential, and so is $A_{\doc}$.
\subsubsection*{Complexity:} 
Let $m_i$ and $n_i$ denote the number of transitions and states of $A_i$, let $\ell = |\doc|$, and $k=|\vars(A_1)\cap\vars(A_2)$. Let $v\df|\vars(A_1)$.
The match structure $M(A_2,\doc)$ can be constructed in $O(\ell n_2^2)$, see~\cite{DBLP:conf/pods/FreydenbergerKP18}.
Recall that $D_2$ has $O(\ell^2 k)$ states and $O(\ell^2 k^2)$ transitions. Each transition can be computed in $O(n_2)$, which means that we can obtain $D_2$ from $M(A_2,\doc)$ in $O(\ell^2 k^2 n_2)$. Hence, the total time of computing $D_2$ from $A_2$ and $\doc$ is $O(\ell^2 k^2 n_2 + \ell n_2^2)$.

To construct an automata $A_{\doc,j}$, we first pre-compute the comparisons of variable configurations in time $O(n_1 k^2)$ (as $D_2$ has $O(k)$ different variable configurations, and variable configurations can be compared in $O(k)$). We also pre-compute the variable-$\emptyword$-closures, which takes $O(m_1 n_1)$.

As $A_1$ has at most $|F_1|$ different skip sets $S_j$, we need to compute $O(n_1)$ sub-automata $A_j$. Each can be constructed with a standard reachability analysis in time $O(m_1+n_1)$. Hence, computing all automata $A_{1,j}$ takes time $O(m_1n_1+n_1^2)$. This will be subsumed by the complexity of the next steps.

For each $A_{\doc,j}$, we have to compute $\ell+1$ levels; in each level, we combine $O(n_1 k)$ state pairs $(p_1,p_2)$ with $O(n_1 k)$ state pairs $(q_1,q_2)$ (as $D_2$ is deterministic and over an alphabet of size $2k$, each state has $O(k)$ outgoing transitions), and we need to include $O(v)$ helper states. Hence, each $A_{\doc,j}$ can be constructed in time $O(\ell k^2 n_1^2 v)$ without pre-computations, and  $O(\ell k^2 n_1^2 v + m_1 n_1)$ including pre-computations. 

Hence, computing $A_{\doc}$ by constructing $O(n_1)$ many $A_{\doc_j}$, each in time $O(\ell k^2 n_1^2 v + m_1 n_1)$, takes a total time of $O(\ell k^2 n_1^3 v + m_1 n_1^2)$.

This combines to a total time of $O(\ell^2 k^2 n_2 + \ell n_2^2 + \ell k^2 n_1^3 v + m_1 n_1^2)$. For a less precise estimation, let $n\in O(n_1+n_2)$ and observe that $k\leq v$ and $m_1\in O(n^2)$. Then the complexity becomes $O(\ell^2 v^2 n + \ell n^2 + \ell v^3 n^3 + n^4)=O(\ell^2 v^2 n  + \ell v^3 n^3 + n^4)$.

\end{document}